\documentclass[envcountsame,envcountsect]{llncs}

\usepackage[T1]{fontenc}
\usepackage{amsmath,amssymb,stmaryrd,wasysym,url}

\usepackage{amsthm}

\usepackage{graphicx}
\usepackage{subcaption}
\usepackage{color}
\usepackage[displaymath, tightpage]{preview}
\usepackage{xypic, float, array,longtable, calrsfs}
\usepackage[mathscr]{eucal}
\usepackage{cancel}
\usepackage[utf8]{inputenc}
\usepackage{dsfont, amscd, latexsym, stmaryrd, amssymb, fancyvrb, ragged2e}
\usepackage{afterpage}

\let\emptyset\varnothing

\newcommand{\lL}{\mathscr{L}}
\newcommand{\Fm}{\mathsf{Form}}

\newcommand{\Nom}{\mathrm{Nom}}
\newcommand{\Prop}{\mathrm{Prop}}
\newcommand{\MOD}{\mathrm{Mod}}
\newcommand{\Act}{\mathrm{Act}}

\newcommand{\mM}{\mathsf{M}}
\newcommand{\Ww}{\mathrm{W}}
\newcommand{\rel}{\mathrm{R}}
\newcommand{\famrpos}{(\rel^+_\pi)_{\pi\in\MOD}}
\newcommand{\famrneg}{(\rel^-_\pi)_{\pi\in\MOD}}
\newcommand{\famrposact}{(\rel^+_a)_{a\in\Act}}
\newcommand{\famrnegact}{(\rel^-_a)_{a\in\Act}}
\newcommand{\rpos}{\rel^+}
\newcommand{\rneg}{\rel^-}
\newcommand{\rboth}{\rel^\pm}
\newcommand{\nval}{\mathrm{N}}
\newcommand{\val}{\mathrm{V}}
\newcommand{\vpos}{\val^+}
\newcommand{\vneg}{\val^-}
\newcommand{\vboth}{\val^\pm}
\newcommand{\four}{\mathbf{4}}

\newcommand{\fequiv}{\equiv_\mathbf{4}}
\newcommand{\irr}{\mathrm{I}}

\newcommand{\cl}{\mathsf{CL}}
\newcommand{\U}{\mathrm{U}}
\newcommand{\same}{\sim_\Theta}
\newcommand{\precededby}{\subseteq_\Theta}

\newcommand{\dbhl}{\textsf{DBHL}}
\newcommand{\fdl}{\textsf{4DL}}
\newcommand{\fhl}{\textsf{4HL}}
\newcommand{\pdl}{\textsf{PDL}}
\newcommand{\ml}{\textsf{ML}}

\newcommand{\redminus}[1]{{\color{red}(}{#1}{\color{red})^-}}
\newcommand{\blueminus}[1]{{\color{blue}(}{#1}{\color{blue})^-}}
\newcommand{\blueneg}[1]{{\color{blue}\neg(}{#1}{\color{blue})}}

\newtheorem{thm}{Theorem}
\newtheorem{cor}{Corollary}
\newtheorem{lem}{Lemma}
\newtheorem{prop}{Proposition}
\newtheorem{defn}{Definition}

\newtheorem{ex}{Example}

\begin{document}
	
	\title{\fdl: a four-valued Dynamic logic and its proof-theory\thanks{This work was partially supported by the Programming Principles, Logic and Verification group at University College London, and by the LASIGE Research Unit, ref.\ UIDB/00408/2020 and ref.\ UIDP/00408/2020}}

	\author{Diana Costa}

	\institute{LASIGE, Faculdade de Ciências, Universidade de Lisboa, PT\\
	\email{dfdcosta@fc.ul.pt}}

	\maketitle		
	
\begin{abstract}
	Transition systems are often used to describe the behaviour of software systems. If viewed as a graph then, at their most basic level, vertices correspond to the states of a program and each edge represents a transition between states via the (atomic) action labelled. In this setting, systems are thought to be consistent and at each state formulas are evaluated as either true or false.
	
	On the other hand, when a structure of this sort -- for example a map where states represent locations, some local properties are known and labelled transitions represent information available about different routes -- is built resorting to multiple sources of information, it is common to find inconsistent or incomplete information regarding what holds at each state, both at the level of propositional variables and transitions.
	
	This paper aims at bringing together Belnap's four values, Dynamic Logic and hybrid machinery such as nominals and the satisfaction operator, so that reasoning is still possible in face of contradicting evidence. \mbox{Proof-theory} for this new logic is explored by means of a terminating, sound and complete tableau system.
\end{abstract}

\section{Introduction}
	In Computer Science, when one is uniquely interested in a program's
	\mbox{input-output} behaviour, one tends to regard computation steps
	as instantaneous, discrete state changes.
	These transition systems are usually viewed as graphs, whose vertices
	correspond to states and labelled edges represent the transitions
	between those states, which result from running a certain program.
	Modelling and reasoning within such systems is commonly done resorting
	to Modal Logics (\ml) \cite{modal_book}, a collection of formal systems
	equipped with operators to express possibility and necessity,
	modalities which are associated with a program, and where the
	valuation of formulas is processed with regard to a specific state.
	The extension to Propositional Dynamic Logic (\pdl) \cite{PDL1,PDL2}
	is particularly well-known for allowing the construction of complex
	programs based on atomic actions.
	Namely, \pdl\ includes sequential programs -- $\alpha;\beta$ --
	to perform first $\alpha$ then $\beta$, choice programs --
	$\alpha\cup\beta$ -- for \mbox{non-deterministically} choosing
	between $\alpha$ or $\beta$, iteration programs -- $\alpha^*$ --
	to run $\alpha$ a \mbox{non-deterministically} chosen number of
	times, and test programs -- $\varphi?$ -- to check that the
	formula $\varphi$ is satisfied before proceeding (keep in mind
	that $\alpha$ and $\beta$ may themselves be complex programs!).
	\pdl\ has proved to be valuable in theoretical CS, Computational Linguistics and Artificial Intelligence.
	
	As ordinary in Modal Logics, formulas in \pdl\ are also
	evaluated locally and can only take one of two values:
	true or false. Classically, as is well-known, the
	concepts ``not true'' and ``false'' are equivalent;
	gluts such as ``true and false'' lead to explosion and
	gaps such as ``neither true nor false'' are not allowed.
	The proposal of Belnap and Dunn consisted in admitting
	these scenarios, therefore allowing each formula to take
	one of four ($\four$) values: (only) true, (only) false,
	neither (true nor false) or both (true and false) \cite{belnap}.
	Belnapian logics are both paraconsistent and paracomplete:
	the Principle of Explosion and the Principle of Excluded Middle
	are both dismissed.
	
	Early attempts at incorporating paraconsistency into Modal,
	Hybrid and Dynamic Logics restricted the use of four values to
	propositional variables \cite{od10,B06,dc16,sedlar}.
	Versions of Modal Logic where both propositional variables
	and accessibility relations are $\mathbf{4}$-valued often seem
	to misbehave in extensions to Hybrid Logic \cite{Blackburn00},
	which in its simplest configuration adds to Modal Logic (i) a new
	set of propositional variables, called nominals, such that each
	holds at a single state, thus naming it, and (ii) the satisfaction operator $@$, with which we can jump to a designated state; namely the value of a formula of the form $@_i\varphi$ is independent from the state where it is being evaluated: the value of $@_i\varphi$ in any state corresponds to the value of $\varphi$ in the state that is
	named by $i$.
	
	The reason to claim that extensions of $\four$-valued
	Modal Logics misbehave has to do precisely with the role
	of nominals in modal formulas.
	Assuming that nominals act classically, \emph{i.e.}, assuming
	that they are two-valued and each holds in a single state as
	described, it would be expected that the value of the
	accessibility relation in the pair $(w,w')$ coincides with
	the value of the formula $@_i\Diamond j$, where $i$ and $j$
	name the states $w$ and $w'$, respectively;
	formally, $\rel(w,w')= \val(@_i\Diamond j,-)$ (note that
	being a satisfaction formula, the state of evaluation
	is irrelevant).
	However, that is not typically the case; an example of a
	Modal Logic that would fail this property in an extension
	to Hybrid Logic can be found in \cite{Riv1}.
	
	Usually Modal Logics where accessibility relations are
	many-valued, such as \cite{Riv1}, preserve the duality
	between modal operators, making the formulas $\Diamond\neg\varphi$
	and $\neg\Box\varphi$ equivalent.
	The approach to negation when it appears directly before
	modal operators seems to be repeatedly the same: rather than
	focusing on the transition, the negation is carried towards
	the non-modal part of the formula, which raises the question:
	if, in a $\four$-valued setting for propositional variables
	and accessibility relations, a propositional variable and its
	negation are evaluated independently, why are not modal formulas
	and their negations?
	
	A core point of this paper is precisely the adoption of
	this alternative approach to the semantics of modal formulas
	and their negations.
	The general idea is that when negation occurs directly before
	a modal operator it \emph{will} directly affect the
	interpretation of the transition.
	In particular, the negation of modal formulas is associated
	with evidence about the \emph{absence} of transitions,
	whereas modal formulas not involving negation before the
	modal operator are interpreted using evidence about the
	\emph{presence} of transitions.
	In \fhl, introduced in Section 2, the formula
	$@_i\neg \Diamond j$ expresses that there is not a
	transition from the state named by $i$ to the state
	named by $j$, whereas the formula $@_i\Box \neg j$
	expresses that all transitions starting from $i$ end
	in a state different from $j$.
	Notice the shift in perspective: the former talks about
	absence of transitions; the latter, in opposition, talks
	about the presence of transitions.
	These formulas will not be regarded as equivalent,
	which means that \fhl\ lacks formal duality (other modal systems that lack formal duality have been proposed in
	\cite{formalduality}).
	
	The presence and absence of transitions is captured by
	decoupling the accessibility relation into positive and
	negative versions, in an analogous fashion to what is
	commonly accepted for propositional variables, with
	positive and negative valuations being used to interpret
	$p$ and $\neg p$, respectively.
	Thus positive accessibility relations will be used
	for the interpretation of formulas of the form $\Diamond\varphi$
	and $\Box\varphi$, and negative accessibility relations
	will be used for the interpretation of formulas of the
	form $\neg\Diamond\varphi$ and $\neg\Box\varphi$.
	
	The logic \fhl\ is endowed with paraconsistency
	and paracompleteness at the level of propositional
	variables and accessibility relations: models in this
	system may satisfy $@_i p$ and $@_i\neg p$ simultaneously,
	or neither, as well as they may satisfy both $@_i \Diamond j$
	and $@_i \neg \Diamond j$, or neither.
	
	Moving forward, the paper presents a dynamic extension of \fhl,
	termed \fdl.
	Section 3 dives into the interpretation of the composition of
	actions, each associated with a \mbox{$\four$-valued} accessibility
	relation, in order to obtain the relations associated with the
	programs of sequence, choice, iteration and test.
	The interpretation of positive relations for each composite program
	follows traditional steps.
	This does not constitute a surprise given that ``positive'' modal
	formulas (where negation does not appear directly before the modal
	operator) are interpreted resorting to the positive accessibility
	relation.
	Curiously, negative relations for composite programs are thought of
	in terms of their complement but are otherwise analogous in construction.
	This strategy of resorting to the complement of the negative relation
	has already been used in the completeness proof of other versions of
	Hybrid Logic with $\four$-valued relations and non-dual modal operators,
	such as \cite{jlamppaper}.
	
	Finally, on Section 4, one can find some proof-theory for \fdl,
	namely under the form of a sound and complete tableau system,
	which combines techniques from \pdl\ \cite{tabpdl1} and Hybrid Logic \cite{hybridproof}.
	
	In summary, one of the central ideas of the paper is that in a
	paraconsistent environment where information about transitions
	may be inconsistent or incomplete, there is an unbreakable link
	between negation and modal operators when they occur in this order.
	This, however, involves a trade-off, as to do so the formal
	duality between modal operators has to be discarded.
	The presence of nominals and of the satisfaction operator is
	crucial in order to enable reference to the state where
	propositional inconsistencies are present, as well as to pinpoint
	which pairs of states have inconsistent transitions.
	Those are the building blocks of a $\four$-valued Hybrid Logic
	which is later extended to a dynamic version where actions are
	composed.
	The latter is accompanied by a sound and complete tableau system.

\section{\fhl: a four-valued Hybrid Logic with non-dual modal operators}

This section explores the logic \fhl, a version of Hybrid Logic with \mbox{$\four$-valued} interpretations of propositional variables and accessibility relations. Two sorts of negation are considered: a classical one, ${\sim}$, so that ${\sim} \varphi$ is read as \emph{it is not the case that $\varphi$}; and a paraconsistent one, $\neg$, such that $\neg\varphi$ is read as \emph{$\varphi$ is false}. In a paraconsistent environment such as \fhl\ the satisfiability of these statements is independent.

The language $\lL$ consists of a set of nominals $\Nom$, a set of propositional variables $\Prop$, a nullary connective $\bot$, unary connectives $\neg,[\pi],\langle\pi\rangle$, where \mbox{$\pi\in\MOD$} and $\MOD$ is a set of modalities, and binary connectives $\wedge,\vee,\rightarrow$. The classical negation of a formula, ${\sim}\varphi$, is defined as $\varphi\rightarrow \bot$, $\top$ is defined as ${\sim}\bot$ and $\varphi\leftrightarrow \psi$ is defined as $(\varphi\rightarrow\psi)\wedge (\psi\rightarrow\psi)$; we will sometimes resort to these abbreviations in order to simplify notation. $\Fm$ is the set of formulas over $\lL$.

\begin{defn}[Model]\label{defn:model}
	A model $\mM$ is a tuple $(\Ww, \famrpos, \famrneg, \nval, \vpos,\linebreak \vneg)$, where:
	
	\begin{itemize}
		\setlength\itemsep{0cm}
		\renewcommand\labelitemi{$\logof$}
		\item $\Ww\neq\emptyset$ is the \emph{domain} of \emph{states} (also known as \emph{worlds});
		
		\item $\rpos_\pi$ and $\rneg_\pi$ are binary relations over $\Ww^2$, called respectively the \emph{positive}/\linebreak\emph{negative} \emph{$\pi$-accessibility relation}, for each modality $\pi$;
		
		\item $\nval:\Nom\rightarrow \Ww$ is a \emph{hybrid nomination}, a function that assigns nominals to elements in $\Ww$; $\nval(i)$ is the element of $\Ww$ \emph{named by} $i$;
		
		\item $\vpos$ and $\vneg$ are \emph{hybrid valuations}, both with domain $\Prop$ and range $\mathscr{P}(\Ww)$, such that $\vpos(p)$ is the set of states where the propositional variable $p$ holds, and $\vneg(p)$ is the set of states where $\neg p$ holds.
	\end{itemize}
\end{defn}

To put it simply, $\vboth(p)$ is the set of states where there is evidence that $p$ is true/false. Analogously, $\rboth_\pi$ is the set of pairs of states for which there is evidence about the existence/absence of a $\pi$-transition. In particular, $\rpos_\pi$ can be seen as a \emph{possibility} relation and $\rneg_\pi$ as an \emph{impossibility} relation, in a similar fashion to $R$ and $R^\sim$ in $\mathrm{FSK^{d-}}$ \cite{formalduality}. Having said that, in \fhl\ the semantics of both $\neg\Diamond\varphi$ and $\neg\Box\varphi$ resorts to the impossibility relation as we shall see next, while in $\mathrm{FSK^{d-}}$ that is only the case for  $\neg\Diamond\varphi$.

\begin{defn}[Satisfaction]\label{defn:satisfaction}
	A satisfaction relation $\models$ between a model $\mM$, a state $w$ and a formula $\varphi$ is defined by structural induction on $\varphi$ in Figure~\ref{fig:satisfaction}.
	
	\begin{figure}[ht]
		\begin{longtable}{rl}
		(i) & $\mM, w\models p \Leftrightarrow w\in \vpos(p)$; $\mM, w\models \neg p \Leftrightarrow w\in \vneg(p)$;\\[0.5ex]
		(ii) & $\mM, w\models i \Leftrightarrow   w=\nval(i)$;	$\mM, w\models \neg i  \Leftrightarrow   w \neq \nval(i)$;\\[0.5ex]
		(iii) & $\mM, w \models\bot$  never;	$\mM, w \models\neg\bot$  always;\\[0.5ex]
		(iv) & $\mM, w \models\neg\neg\varphi  \Leftrightarrow   \mM, w \models\varphi$;\\[0.5ex]
		(v) & $\mM, w \models \varphi\wedge\psi \Leftrightarrow   \mM, w \models\varphi\ \text{and}\ \mM, w \models \psi$;\\[0.5ex]
		& $\mM, w\models \neg (\varphi\wedge \psi)  \Leftrightarrow   \mM, w\models\neg\varphi\ \text{or}\ \mM, w \models \neg\psi$;\\[0.5ex]
		(vi) & $\mM, w \models \varphi\rightarrow\psi \Leftrightarrow   \mM, w \models\varphi\ \text{implies}\ \mM, w \models \psi$;\\[0.5ex]
		& $\mM, w\models \neg (\varphi\rightarrow \psi)  \Leftrightarrow   \mM, w\not\models\neg\varphi\ \text{and}\ \mM, w \models \neg\psi$;\\[0.5ex]
		(vii) & $\mM, w \models \varphi\vee \psi \Leftrightarrow  \mM, w \models\varphi\ \text{or}\ \mM, w \models\psi$\\[0.5ex]
		& $\mM, w \models \neg(\varphi\vee \psi) \Leftrightarrow  \mM, w \models\neg \varphi\ \text{and}\ \mM, w \models\neg\psi$\\[0.5ex]
		(viii) & $\mM, w \models \langle\pi\rangle \varphi  \Leftrightarrow   \exists w'(w\rpos_\pi w'\ \text{and } \mM,w'\models \varphi)$;\\[0.5ex]
		& $\mM, w \models \neg\langle\pi\rangle \varphi \Leftrightarrow 
		 \forall w'(w\rneg_\pi w' \text{ or } \mM,w'\models \neg\varphi)$\\[0.5ex]
		& $\phantom{\mM, w \models \neg\langle\pi\rangle \varphi} \Leftrightarrow
		\forall w'(w(\rneg_\pi)^c w' \text{ implies } \mM,w'\models \neg\varphi)$;\\[0.5ex]
		(ix) & $\mM, w \models [\pi] \varphi \Leftrightarrow   \forall w'(w\rpos_\pi w'\ \text{implies } \mM, w'\models \varphi)$;\\[0.5ex]
		& $\mM, w \models \neg[\pi] \varphi \Leftrightarrow   \exists w'(\text{it is false}(\mM,w'\models \neg \varphi \text{ implies } w\mathrm{R}^-_\pi w'))$\\[0.5ex]
		& $\phantom{\mM, w \models \neg[\pi] \varphi} \Leftrightarrow   \exists w'(w(\mathrm{R}^-_\pi)^c w'\ \text{and } \mM,w'\models \neg \varphi)$;\\[0.5ex]
		(x) & $\mM, w \models @_i \varphi \Leftrightarrow   \mM, w' \models \varphi, \text{where}\ w'=\nval(i)$;\\[0.5ex]
		& $\mM, w \models \neg @_i \varphi \Leftrightarrow   \mM, w' \models \neg\varphi,  \text{where}\ w'=\nval(i)$.\\[0.5ex]
		\end{longtable}
		\caption{Definition of the satisfaction relation $\mM,w\models\varphi$ for \fhl.\label{fig:satisfaction}}
	\end{figure}

We say that a formula $\varphi$ is \emph{globally satisfied} if $\mM\models\varphi$, \emph{i.e.}, $\mM, w \models \varphi$ for all $w\in \Ww$; a formula is \emph{valid}, $\models \varphi$, if it is globally satisfied in all models.
Two formulas are equivalent, denoted $\varphi\equiv \psi$, if for every model $\mM$ and every state $w$, $\mM,w\models \varphi$ if and only if $\mM,w\models \psi$; two formulas are $\mathbf{4}$-equivalent, denoted $\varphi\fequiv\psi$ if $\varphi\equiv \psi$ and $\neg\varphi\equiv \neg\psi$.
\end{defn}

We read $\mM, w\models \varphi/\neg\varphi$ as ``in the model $\mM$ there is information that at state $w$ the formula $\varphi$ is true/false''. In terms of notation, $(\rboth_{\pi})^c$ denotes the complement of the relation $\rboth_\pi$ in $\Ww^2$.

Observe that the formal duality between modal operators does not hold for the paraconsistent negation $\neg$. A quick look at the semantics in Figure~\ref{fig:satisfaction} is enough to check that $\neg\langle\pi\rangle\varphi\not\equiv [\pi]\neg \varphi$ and $\langle\pi\rangle\neg \varphi\not\equiv \neg[\pi] \varphi$. The reason behind this is the fact that the semantics for what we will call \emph{positive} modal formulas -- formulas of the form $\langle\pi\rangle\varphi$ or $[\pi]\varphi$ -- relies on evidence about the presence of transitions, whereas for \emph{negative} modal formulas -- formulas of the form $\neg \langle\pi\rangle\varphi$ or $\neg [\pi]\varphi$ -- it relies on evidence about the absence of transitions. In \fhl\ modal formulas are interpreted roughly as follows:

\begin{itemize}
	\setlength\itemsep{0cm}
	\renewcommand\labelitemi{$\logof$}
	\item $\langle\pi\rangle \varphi$ is intuitively interpreted as \emph{it is possible to reach $\varphi$ at least true}.
	
	The formula $@_i\langle\pi\rangle j$ holds in a model if and only if there is evidence of a $\pi$-transition from the state named by the nominal $i$ to the state named by the nominal $j$, \emph{i.e.}, $\nval(i)\rpos_\pi \nval(j)$.
	
	\item $[\pi]\varphi$ is intuitively interpreted as \emph{it is mandatory to reach $\varphi$ at least true}.
	
	The formula $@_i[\pi]\neg j$ holds if and only if all evidence about the existence of $\pi$-transitions from the state named by the nominal $i$ lead towards a state $w'$ which is not named by $j$; therefore $\nval(i)(\rpos_\pi)^c \nval(j)$.

	\item $\neg\langle\pi\rangle \varphi$ is intuitively interpreted as \emph{it is not possible to reach $\varphi$ at most true}.
	This formula holds when, for all states $w$ where $\neg \varphi$ does not hold (\emph{i.e.} where $\varphi$ may be at most true) there is evidence about the absence of a $\pi$-transition from the current state to $w$. 
	
	The formula $@_i\neg \langle\pi\rangle j$ holds if and only if there is evidence that there is no $\pi$-transition from the state named by the nominal $i$ to the state named by the nominal $j$, \emph{i.e.}, $\nval(i)\rneg_\pi \nval(j)$.
	
	\item $\neg[\pi]\varphi$ is intuitively interpreted as \emph{it is not necessary to reach $\varphi$ at most true}. This formula holds if there are states where, even though $\varphi$ is false, there is no evidence that the $\pi$-transition from the current state towards there is missing.
	
	The formula $@_i\neg [\pi]\neg j$ holds if and only if there is no evidence that there is no $\pi$-transition from the state named by the nominal $i$ to the state named by the nominal $j$, \emph{i.e.}, $\nval(i)(\rneg_\pi)^c \nval(j)$.
\end{itemize}

It is however the case that ${\sim}\langle\pi\rangle\varphi\fequiv [\pi]{\sim} \varphi$ and $\langle\pi\rangle\varphi\fequiv {\sim}[\pi]{\sim} \varphi$ (keep this observation in mind as it will play an important role soon). Furthermore, even though the accessibility relations are $\four$-valued, \fhl\ is a normal logic, \emph{i.e.} the \mbox{$\mathrm{K}$-axiom} $[\pi](\varphi\rightarrow\psi)\rightarrow ([\pi]\varphi \rightarrow [\pi] \psi)$ is valid.

In general, $\not\models{\sim}\varphi \rightarrow \neg\varphi$ and $\not\models\neg\varphi \rightarrow {\sim}\varphi$, since the notions ``\emph{not true}'' and ``\emph{false}'' are independent in this framework. The exception are pure formulas, \emph{i.e.} those which involve only nominals, which still behave classically. Also, in general, ${\sim}\neg\varphi\not\equiv \varphi$ (if it is not the case that $\varphi$ is false, then $\varphi$ is not necessarily true) and $\neg{\sim}\varphi\not\equiv \varphi$; nonetheless $\neg{\sim}\varphi\fequiv {\sim}\neg\varphi$.

Let us at last briefly address the choice for the connective $\rightarrow$. While $\rightarrow$ is neither the strong implication nor the weak implication found in literature (\cite{wansing,Riv1} among many others), it appears very naturally here. Note that \mbox{$\mM, w\models \varphi \rightarrow \psi$} is interpreted as an ``actual'' implication, \emph{i.e.}, if and only if $\mM, w\models \varphi$ \emph{implies} $\mM, w\models \psi$. This is not an unusual approach; in fact, it mimics weak implication. Given that the ``positive'' implication is commonly accepted, observe that \mbox{$\mM, w \models {\sim}\varphi$} if and only if $\mM, w \not\models \varphi$. Thus $\varphi\rightarrow \psi \equiv {\sim}\varphi \vee \psi$. If we force the paraconsistent negation on both sides, we get $\neg (\varphi\rightarrow \psi) \equiv \neg ({\sim}\varphi \vee \psi)$, which by simple inspection of the semantics and the previous comments about $\sim$ and $\neg$ yields that \mbox{$\neg (\varphi\rightarrow \psi) \equiv {\sim}\neg\varphi \wedge \neg\psi$}.

\smallskip

The use of both ${\sim}$ and $\neg$ allows us to express Belnap's four truth-values of a formula $\varphi$ in a model $\mM$ easily:

\begin{itemize}
	\item $\varphi$ is only true at $w$: $\mM, w\models \varphi \wedge {\sim}\neg\varphi$;
	\item $\varphi$ is only false at $w$: $\mM, w\models \neg\varphi \wedge {\sim}\varphi$;
	\item $\varphi$ is both true and false at $w$: $\mM, w\models \varphi \wedge \neg\varphi$;
	\item $\varphi$ is neither true nor false at $w$: $\mM, w\models {\sim}\varphi \wedge {\sim}\neg\varphi$.
\end{itemize}

The presence of the satisfaction operator and nominals is enough to make it possible to describe named models (where all states are named by at least one nominal) in \fhl\ by the set of basic formulas that hold there, a method known as \emph{diagram} in classical Hybrid Logic. Since there is no direct relationship between the satisfiability of a propositional variable and its negation and between evidence about the presence and absence of transitions, diagrams are required to contain more than (Hybrid Logic's) atoms. Formally:

\begin{defn}\label{defn:diagram}
	Consider the set of irreducible information $\irr =\{@_i p, @_i \neg p,\linebreak @_i\langle \pi\rangle j, @_i \neg \langle \pi\rangle j, @_i j \ |\ i,j\in\Nom,~p\in\Prop,~\pi\in\MOD\}$.
	The diagram of a named model $\mM$ is the following set:

	$$\begin{array}{rcl}
	\mathrm{Diag}(\mM) & = & \{ a\in \irr \text{ s.t. } \mM \models a\}.
	\end{array}$$
\end{defn}

The following example illustrates this notion:

\begin{ex}\label{ex:diagram}
	Consider the model $\mM=(\Ww=\{w_1,w_2,w_3,w_4,w_5\},\linebreak \rpos=\{(w_1,w_2), (w_4,w_3)\},\ \rneg=\{(w_1,w_2), (w_1,w_3)\},\ \nval: (i\mapsto w_1, j\mapsto w_2,\linebreak k\mapsto w_3, l\mapsto w_4, m\mapsto w_5),\ \vpos: (p\mapsto \{w_2,w_4\}, q\mapsto \emptyset),\ \vneg: (p\mapsto \{w_4\},\linebreak q\mapsto \{w_3\}) )$, represented in Figure~\ref{fig:diagram}.
	
	\begin{figure}[ht]
		\centering
		\includegraphics[width=0.5\textwidth]{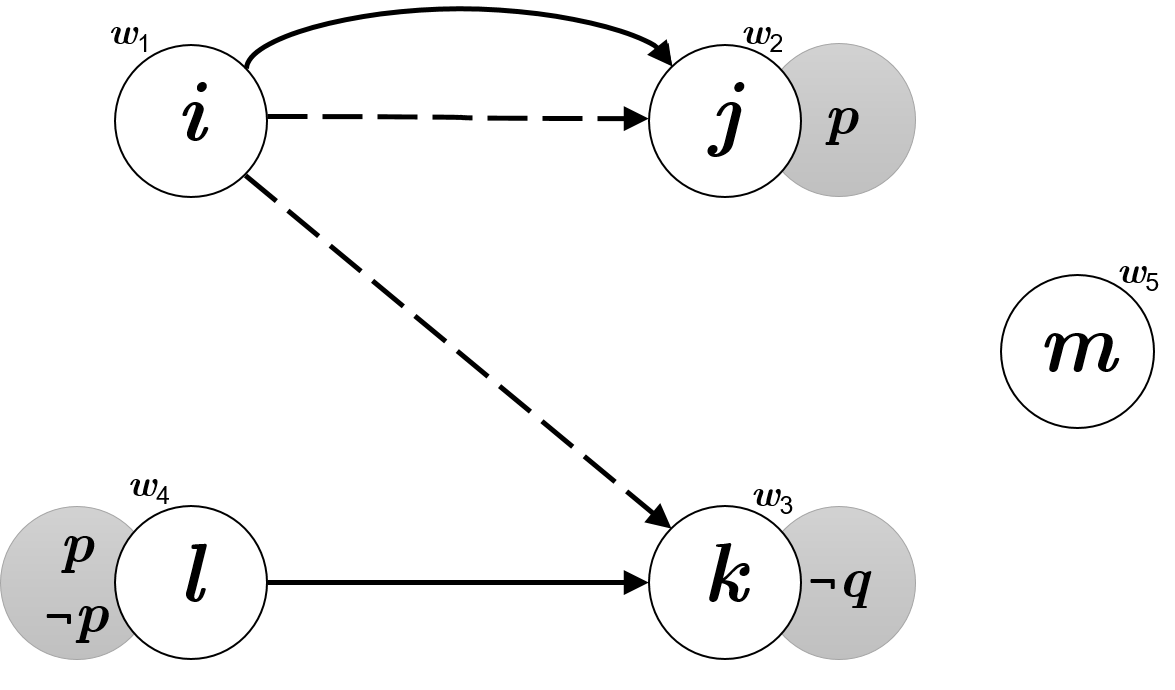}
		\caption{A sketch of model $\mM$: full lines indicate info. about the presence of a transition, dashed lines indicate info. about the absence of a transition.}\label{fig:diagram}
	\end{figure}

	The diagram of $\mM$ comes as follows:
	
	$$\begin{array}{rcll}
	\mathrm{Diag}(\mM) & = & \{@_i \Diamond j, @_l \Diamond k, & \text{info. about presence of transitions}\\
	& & @_i \neg \Diamond j, @_i \neg \Diamond k, & \text{info. about absence of transitions}\\
	& & @_j p, @_k \neg q, @_l p, @_l \neg p, & \text{info. about local properties}\\
	& & @_i i, @_j j, @_k k, @_l l, @_m m\} & \text{info. about equalities between states}
	\end{array}$$
\end{ex}

We say that $\psi$ is a global consequence of a set of formulas $\Gamma$, denoted $\Gamma \models_g \psi$, if for all models $\mM$, if $\mM\models \Gamma$ then $\mM\models \psi$.
We say that $\psi$ is a local consequence of a set of formulas $\Gamma$, denoted $\Gamma \models_\ell \psi$, if for every model $\mM$ and state $w$, if $\mM,w \models \Gamma$ then $\mM,w \models \psi$.
As is the case for the logic $\mathsf{BK}$ (\cite{wansing}), the logic \fhl\ does not enjoy the Deduction theorem for global consequence, \emph{i.e.}, it is not the case that $\varphi\models_g \psi$ if and only if $\models \varphi\rightarrow \psi$. However, note that both $\mathsf{BK}$ and \fhl\ enjoy the Deduction theorem for local consequence. For the remainder of the paper, whenever it is mentioned that a formula is a consequence of a set of formulas, we will be referring to the notion of \emph{global} consequence.

The counter-example that serves to prove that $\mathsf{BK}$ is not closed under the replacement rule is not a counter-example for \fhl\ (of course, it must be adapted to the semantics in \fhl; the details are below).
The replacement rule says that if $\varphi \leftrightarrow \psi$ is valid, then the formula $\chi(\varphi)\leftrightarrow \chi(\psi)$ is valid as well.
Both of the formulas $\neg (\varphi\rightarrow \psi)\leftrightarrow ({\sim}\neg\varphi \wedge \neg\psi)$ and \mbox{$\neg \neg (\varphi\rightarrow \psi) \leftrightarrow \neg (\sim\neg\varphi \wedge \neg\psi)$} are valid in \fhl.
Clearly, this is a bi-product of the choice for the semantics of $\rightarrow$ that has been made in this paper.
Indeed, \fhl\ is closed under the replacement rule. It is easier to check it under a $\mathbf{4}$-valued approach, so the details can be found in the next subsection.

Recall also that, being a paraconsistent and paracomplete logic, the schemas $(\varphi\wedge \neg\varphi)\rightarrow\bot$ and $\varphi \vee\neg\varphi$ are not valid. Nonetheless, the classical behaviour of the negation ${\sim}$ makes it obvious that $(\varphi\wedge {\sim}\varphi)\rightarrow\bot$ and $\varphi \vee{\sim}\varphi$ are both valid schemas.

\subsection*{A direct $\mathbf{4}$-perspective of \fhl}

An alternative definition of model uses functions $\rel_\pi: \Ww\times \Ww \rightarrow \mathbf{4}$ and\linebreak \mbox{$\val: (\Prop\cup\Nom)\times \Ww\rightarrow \mathbf{4}$}, where $\mathbf{4}=\{\mathrm{t}, \mathrm{f}, \mathrm{b}, \mathrm{n}\}$ is the set of Belnapian truth values: true, false, both true and false, and neither true nor false, respectively.
We say that an element $x\in\four$ yields positive evidence if and only if $x\in\{\mathrm{t},\mathrm{b}\}$ and that it yields negative evidence if and only if $x\in\{\mathrm{f},\mathrm{b}\}$.
Formally:

\begin{defn}[$\mathbf{4}$-model]\label{def:4structure}
	A \emph{$\mathbf{4}$-model} $\overline{\mM}$ is a tuple $(\Ww, \left(\rel_\pi \right)_{\pi\in\MOD}, \val)$, where:
	
	\begin{itemize}
		\setlength\itemsep{0cm}
		\renewcommand\labelitemi{$\logof$}
		\item $\Ww\neq\emptyset$ is the domain;
		\item $\rel_\pi$ is an \emph{accessibility function} such that $\rel_\pi:\Ww\times \Ww\rightarrow \mathbf{4}$ for each $\pi\in\MOD$;
		\item $\val$ is a \emph{valuation} with domain $(\Prop\cup\Nom)\times \Ww$ and range $\mathbf{4}$ such that $\val(i,w)=\mathrm{t}$ for a unique $w\in \Ww$ and $\val(i,w')=\mathrm{f}$ for every other state $w'$.
	\end{itemize}
\end{defn}

Belnap's four values may be arranged according to two partial orders: the first one, $\leq_\mathsf{t}$, reflects the ``\emph{quality}'' of the information, whereas the second, $\leq_\mathsf{k}$, reflects the ``\emph{quantity}'' of information. The bilattice structure is represented in Figure~\ref{fig:bilattices}.

\begin{figure}[H]
	\centering
	\includegraphics[width=0.3\textwidth]{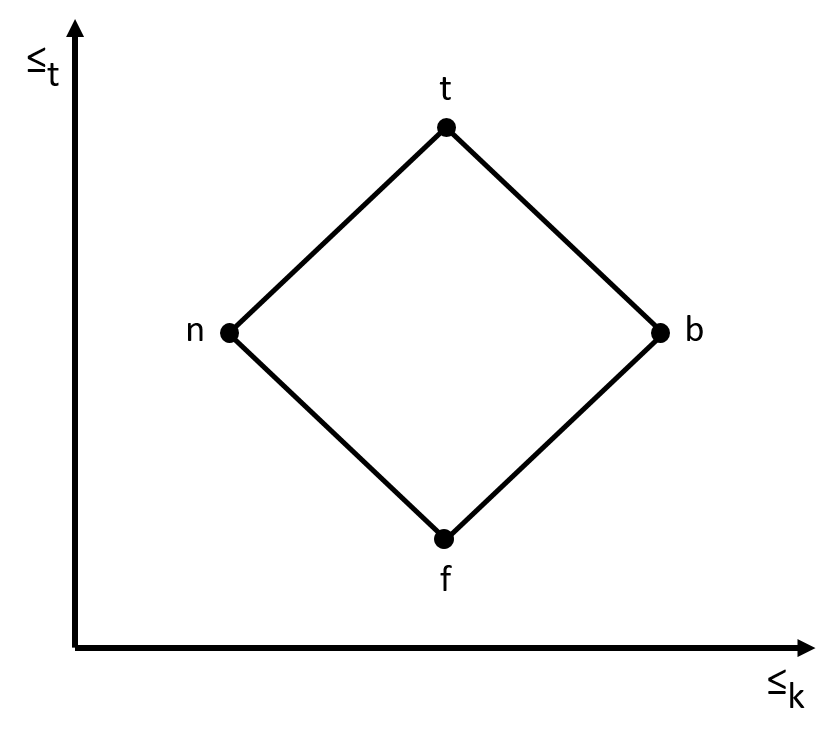}
	\caption{The two orderings of Belnap's bilattice.\label{fig:bilattices}}
\end{figure}

\begin{defn}[Satisfaction in $\mathbf{4}$-models]\label{def:sat_4struc}
	A satisfaction relation $\models_\mathbf{4}$ between a $\mathsf{4}$-model $\overline{\mM}$, a state $w$ and a formula $\varphi$ is defined as follows:
	
	\[\overline{\mM}, w \models_\mathbf{4} \varphi \Leftrightarrow \val(\varphi,w)\in\{\mathrm{t},\mathrm{b}\}\]
	
	\noindent where $\val$ is extended to all formulas as in Figure~\ref{fig:extended_val}.
	
	Globally satisfied and valid formulas are defined as in Definition~\ref{defn:satisfaction}.
\end{defn}

\begin{figure}[h!t]
	\centering
$\begin{array}{l}
	\val(\bot,w) =\mathrm{f}\\[0.5ex]
	
	\val(\neg \varphi,w) =\neg \val(\varphi,w), \text{ where } \neg \mathrm{t}=\mathrm{f},\ \neg \mathrm{f}=\mathrm{t},\ \neg \mathrm{b}=\mathrm{b} \text{ and } \neg \mathrm{n}=\mathrm{n}\\[0.5ex]
	
	\val({\sim} \varphi,w) ={\sim} \val(\varphi,w), \text{ where } {\sim} \mathrm{t}=\mathrm{f},\ {\sim} \mathrm{f}=\mathrm{t},\ {\sim} \mathrm{b}=\mathrm{n} \text{ and } {\sim} \mathrm{n}=\mathrm{b}\\[0.5ex]
	
	\val(\varphi\wedge\psi,w) =  \val(\varphi,w)\varowedge \val(\psi,w) = \underset{\leq_\mathsf{t}}{\mathrm{inf}}\ \{\val(\varphi,w), \val(\psi,w)\}\\[0.5ex]
	
	\val(\varphi\vee\psi,w)= \val(\varphi,w)\varovee \val(\psi,w) = \underset{\leq_\mathsf{t}}{\mathrm{sup}}\{\val(\varphi,w), \val(\psi,w)\}\\[0.5ex]
	
	\val(\varphi\rightarrow\psi,w) = \val(\varphi,w) \rhd \val(\psi,w) = {\sim}\val( \varphi,w) \varovee \val(\psi,w)\\[0.5ex]
	
	\val(@_i \varphi,w) = \val(\varphi,w'), \text{ where } w' \text{ is such that } \val(i,w')=\mathrm{t}\\[0.5ex]
	
	\val(\langle\pi\rangle\varphi,w) = \underset{\leq_\mathsf{t}}{\mathrm{sup}} \{\rel_\pi(w,w') \varowedge \val(\varphi,w'),\ w'\in \Ww\}\\[0.5ex]
	
	\val([\pi]\varphi,w) = \underset{\leq_\mathsf{t}}{\mathrm{inf}} \{{\sim} (\rel_\pi(w,w') \varowedge \val({\sim}\varphi,w')),\ w'\in \Ww\}\\
	\hphantom{\val([\pi]\varphi,w)} = \underset{\leq_\mathsf{t}}{\mathrm{inf}} \{{\sim}\rel_\pi(w,w') \varovee \val(\varphi,w'),\ w'\in \Ww\}\\
	\hphantom{\val([\pi]\varphi,w)} = \underset{\leq_\mathsf{t}}{\mathrm{inf}} \{\rel_\pi(w,w') \rhd \val(\varphi,w'),\ w'\in \Ww\}
\end{array}$
\caption{Extension of a $\mathbf{4}$-valuation to formulas.}\label{fig:extended_val}
\end{figure}

The meet and join in $(\four,\leq_\mathsf{t})$ yield a pessimistic, respectively optimistic, approach to evidence.
The meet of elements $x,y \in \four$ yields positive evidence if both $x$ and $y$ yield positive evidence, and yields negative evidence if either $x$ or $y$ does. Conversely the join yields positive evidence if either $x$ or $y$ does, and yields negative evidence if both $x$ and $y$ do.
This corresponds to the approach adopted for conjunction and disjunction in $\fhl$.

The representation of elements $x\in\four$ as pairs $(a,b)\in \{0,1\}^2$ such that $a=1$ iff $x\in\{\mathrm{t},\mathrm{b}\}$ (\emph{i.e.}, the first projection retains positive evidence) and $b=1$ iff $x\in\{\mathrm{f},\mathrm{b}\}$ (\emph{i.e.}, the second projection retains negative evidence) makes the behaviour of these operations clearer.

Additionally, when viewed as a pair, the paraconsistent negation in $\fhl$, $\neg$, of an element in $\four$ switches the positions of positive and negative evidence, (\emph{i.e.}, $\neg(a,b)=(b,a)$); the classical negation, $\sim$, inverts the value in each position of the pair (\emph{i.e.}, ${\sim}(a,b) = (1-a,1-b)$). In \cite{defaultneg}, $\neg$ is simply referred to as ``negation'', while $\sim$ is called ``default negation''.

This direct approach to satisfaction puts in evidence the spontaneity of the choice for the semantics of $\rightarrow$-formulas. First, the semantics of diamond-formulas is structurally the same as in the classical case. The semantics of box-formulas is constructed resorting to the equivalence $[\pi]\varphi\equiv {\sim}\langle\pi\rangle{\sim}\varphi$ (briefly hinted at previously). By analogy with the classical structure of the semantics of box-formulas, the use of a new implication (as mentioned before, it differs from the $\mathbf{4}$-valued weak and strong implications in \cite{wansing,Riv1}) is thus justified. Obviously, $\rightarrow$ and $\wedge$ constitute a residuated pair in $\mathbf{4}$ with respect to $\leq_\mathsf{t}$.

Going back to the observation about the replacement rule, note that valid formulas have value $\mathsf{t}$. A formula $\varphi \leftrightarrow \psi$ has value $\mathsf{t}$ if and only if the formulas $\varphi$ and $\psi$ have exactly the same value (which is not the case for either weak nor strong implication). Therefore, it is expected that $\chi(\varphi)\leftrightarrow \chi(\psi)$ is valid as well.

Observe that for nominals $i$ and $j$ and an arbitrary state $w$,

\[\val(@_i\langle\pi\rangle j, w)=\rel_\pi(w',w''),\]

\noindent where $w'$ and $w''$ are such that $\val(i,w')=\val(j,w'')=\mathrm{t}$.

\begin{defn}[Equivalent models]\label{defn:equiv_models}
	A model $\mM$ and a $\mathbf{4}$-model $\overline{\mM}$ are equivalent when their domains coincide and for all $w,w'\in \Ww$, $\pi\in\MOD$, $p\in\Prop$, \mbox{$i\in\Nom$}, the equivalences in Figure~\ref{fig:equiv_models} hold.

\begin{figure}[h!t]
	\centering
	$\begin{array}{lll}
		\rel_\pi(w,w')\in\{\mathrm{t},\mathrm{b}\} & \iff & (w,w')\in \rpos_\pi\\[0.5ex]
		\rel_\pi(w,w')\in\{\mathrm{f},\mathrm{b}\} & \iff & (w,w')\in \rneg_\pi\\[0.5ex]
		\val(p,w)\in\{\mathrm{t},\mathrm{b}\} & \iff & w\in \vpos(p)\\[0.5ex]
		\val(p,w)\in\{\mathrm{f},\mathrm{b}\} & \iff & w\in \vneg(p)\\[0.5ex]
		\val(i,w)=\mathrm{t} & \iff & w=N(i)\\[0.5ex]
		\val(i,w)=\mathrm{f} & \iff & w\neq N(i)\\
	\end{array}$
	\caption{Requirements for the equivalence between a $\mathbf{4}$-model and a model.}\label{fig:equiv_models}
\end{figure}
\end{defn}

The following result holds:

\begin{lem}\label{lem:equiv_semantics}
	For a model $\mM$ and a $\mathbf{4}$-model $\overline{\mM}$ such that $\mM$ and $\overline{\mM}$ are equivalent, it is the case that:
	
	$$\overline{\mM}, w \models_\mathbf{4} \varphi \Leftrightarrow \mM, w\models \varphi,\text{ for all } \varphi\in\Fm.$$
\end{lem}

\fhl\ is closely related to the Double-Belnapian Hybrid Logic \dbhl$*$, for which there is a terminating, sound and complete tableau system \cite{jlamppaper}. Even though we will not provide details, we can construct an analogous system for \fhl\ by updating some of the rules of the calculus for \dbhl$*$.

In a nutshell, there are three key aspects that distinguish between \fhl\ and \dbhl$*$, namely:
\begin{enumerate}
	\item In \dbhl$*$ the implication used is weak implication \cite{wansing,Riv1}; there is no room for the definition of a classical negation built from implication, as we do have in \fhl with $\sim$.
	\item In \dbhl$*$ the semantics for disjunction involves the so-called disjunctive syllogism, meaning that for a disjunction to hold, not only one of the disjuncts has to, but also in case the negation of a disjunct holds, then the other disjunct must hold as well.
	\item The major differentiating aspect is the semantics for negative modal formulas, although both systems agree on the semantics for positive modal formulas. In \fhl\ and \dbhl$*$, the semantics of positive modal formulas requires the non-modal part of the formula to be at least true (meaning that the non-modal part could be evaluated as (only) true or both (true and false)). In \fhl, for negative modal formulas it is required that the non-modal part is at most true (meaning that it could be evaluated as (only) true or neither (true nor false)). On the other hand, in \dbhl$*$ the provision for negative modal formulas is that the non-modal part is also at least true.
\end{enumerate}

\section{\fdl: a dynamic extension of \fhl}

Usually what we call modalities in Modal and Hybrid Logics is designated as atomic actions in a Dynamic Logic environment. We will follow that terminology in this section and denote the set of atomic actions by $\Act$ and its elements by $a, b, c,$ \emph{etc}.
The language $\lL_{dyn}$ is a variant of the language of \fhl, containing two kinds of expressions, namely, programs $\alpha$ and formulas $\varphi$ as follows:

$$\begin{array}{rl}
\alpha::= & a \mid \alpha;\alpha \mid \alpha\cup\alpha \mid \alpha^* \mid \varphi?\\
\varphi::= & p \mid i \mid \bot \mid \neg\varphi \mid \varphi\wedge\varphi \mid \varphi\vee\varphi \mid \varphi\rightarrow\varphi \mid [\alpha]\varphi \mid \langle\alpha\rangle\varphi
\end{array}$$

\noindent where $a\in\Act$, $p\in\Prop$, $i\in\Nom$.

The program $\alpha;\beta$ is a sequential program that runs first $\alpha$, then $\beta$. The program $\alpha\cup\beta$ non-deterministically chooses between $\alpha$ and $\beta$. The program $\alpha^*$ iteratively runs $\alpha$ a non-deterministically chosen number of times. The test program $\varphi?$ verifies if $\varphi$ holds at the current state and if so, proceeds.

\begin{defn}[Program interpretation]\label{defn:programs_interpr}
	The positive interpretations of composite programs are built as follows:
	
	\begin{itemize}
		\setlength\itemsep{0cm}
		\renewcommand\labelitemi{$\logof$}
		\item $\rpos_{\alpha;\beta}$ is the composition of $\rpos_\alpha$ and $\rpos_\beta$;
		\item $\rpos_{\alpha\cup\beta}$ is the union of $\rpos_\alpha$ and $\rpos_\beta$;
		\item $\rpos_{\alpha^*}$ is the reflexive transitive closure of $\rpos_\alpha$;
		\item $\rpos_{\varphi?}$ is defined as the set $\{(w,w)\ |\ \mM, w\models\varphi\}$.
	\end{itemize}
	
	As for negative interpretations, we define their complements; the following constructions are actually analogous to the previous ones:
	
	\begin{itemize}
		\setlength\itemsep{0cm}
		\renewcommand\labelitemi{$\logof$}
		\item $(\rneg_{\alpha;\beta})^c$ is the composition of $(\rneg_\alpha)^c$ and $(\rneg_\beta)^c$;
		\item $(\rneg_{\alpha\cup\beta})^c$ is the union of $(\rneg_\alpha)^c$ and $(\rneg_\beta)^c$;
		\item $(\rneg_{\alpha^*})^c$ is the reflexive transitive closure of $(\rneg_\alpha)^c$;
		\item $(\rneg_{\varphi?})^c$  is defined as the set $\{(w,w)\ |\ \mM, w\not\models \neg\varphi\}$.
	\end{itemize}
	
\end{defn}

	It is straightforward to observe that the positive relations and the complement of negative relations associated with composite programs $;,\cup,-^*$ are built in exactly the same way as in classical \pdl.
	
	The positive and negative versions of the relation associated with the test program however, hold a closer relation. These types of relation are not defined with respect to the accessibility relations, but instead depend on the local satisfaction of a particular formula. Thus note that $(\rneg_{\varphi?})^c= \{(w,w)\ |\ w\in \Ww\}\backslash \rpos_{(\neg\varphi)?}$.
	
	Semantics is easily adapted from Definition \ref{defn:satisfaction} to deal with the use of composite programs in place of the modalities in modal formulas. So, for a model $(\Ww, \famrposact, \famrnegact, \nval, \vpos, \vneg)$ (as in Definition~\ref{defn:model}, except that we update the terminology and use $\Act$ instead of $\MOD$), semantics of modal formulas is defined as one would expect:
	
	\begin{itemize}
	\item $\mM, w \models \langle\alpha\rangle \varphi  \Leftrightarrow   \exists w'(w\rpos_\alpha w'\ \text{and } \mM,w'\models \varphi)$;
	\item $\mM, w \models \neg\langle\alpha\rangle \varphi \Leftrightarrow    \forall w'(w(\rneg_\alpha)^c w' \text{ implies } \mM,w'\models \neg\varphi)$;
	\item $\mM, w \models [\alpha] \varphi \Leftrightarrow   \forall w'(w\rpos_\alpha w'\ \text{implies } \mM, w'\models \varphi)$;
	\item $\mM, w \models \neg[\alpha] \varphi \Leftrightarrow   \exists w'(w(\mathrm{R}^-_\alpha)^c w'\ \text{and } \mM,w'\models \neg \varphi)$.
	\end{itemize}

	The interpretation of composite programs is such that the axioms of \pdl\ hold in \fdl, making it a \emph{truly Dynamic Logic}.
	
	\begin{lem}\label{lem:dynamic_axioms}
	The following schemes hold in every model:
	\begin{enumerate}
		\item $[\alpha;\beta]\varphi \fequiv [\alpha][\beta]\varphi$
		\item $[\alpha\cup\beta]\varphi \fequiv [\alpha]\varphi \wedge [\beta]\varphi$
		\item $[\varphi?]\psi \fequiv \varphi\rightarrow\psi$
		\item $[\alpha^*]\varphi \fequiv \varphi\wedge ([\alpha][\alpha^*]\varphi)$
	\end{enumerate}

	The derivations with $\langle-\rangle$ can also be verified: 
	\begin{enumerate}
		\item[1'.] $\langle\alpha;\beta\rangle \fequiv \langle \alpha\rangle\langle\beta\rangle \varphi$
		\item[2'.] $\langle\alpha\cup\beta\rangle\varphi \fequiv \langle\alpha\rangle \varphi \vee \langle\beta\rangle\varphi$
		\item[3'.] $\langle\varphi?\rangle \psi\fequiv \varphi\wedge \psi$
		\item[4'.] $\langle\alpha^*\rangle\varphi \fequiv \varphi \vee (\langle\alpha\rangle\langle\alpha^*\rangle \varphi)$
	\end{enumerate}
	\end{lem}

	\begin{proof}	
		Let $\mM$ be a model and $w$ a state in $\Ww$. Then:
		\begin{enumerate}
			\item[1.] $\vphantom{\ }$
			
			$\begin{array}{rcl}
			\mM, w \models [\alpha;\beta]\varphi & \iff & \forall w'(w\rpos_{\alpha;\beta}w' \text{ implies } \mM, w'\models \varphi)\\
			& \iff & \forall w'(\exists u (w\rpos_\alpha u \text{ and } u\rpos_\beta w') \text{ implies } \mM, w' \models \varphi)\\
			& \iff & \forall w'(\forall u (w(\rpos_\alpha)^c u \text{ or } u(\rpos_\beta)^c w') \text{ or } \mM, w' \models \varphi)\\
			& \iff & \forall u,w' (w (\rpos_\alpha)^c u \text{ or } u (\rpos_\beta)^c w' \text{ or } \mM, w' \models \varphi)\\
			& \iff & \forall u (w (\rpos_\alpha)^c u \text{ or } \forall w'(u (\rpos_\beta)^c w' \text{ or } \mM, w' \models \varphi))\\
			& \iff & \forall u (w\rpos_\alpha u \text{ implies } \forall w'(u \rpos_\beta w' \text{ implies } \mM, w' \models \varphi))\\
			& \iff & \forall u (w\rpos_\alpha u \text{ implies } \mM, u\models[\beta]\varphi)\\
			& \iff & \mM, w \models [\alpha][\beta]\varphi\\
			& & \\
			& & \\
			\mM, w \models \neg[\alpha;\beta]\varphi & \iff & \exists w' (w (\rneg_{\alpha;\beta})^c w' \text{ and } \mM, w'\models \neg\varphi)\\
			& \iff & \exists w' (\exists u (w (\rneg_{\alpha})^c u \text{ and } u (\rneg_{\beta})^c w') \text{ and } \mM, w'\models \neg\varphi)\\
			& \iff & \exists u, w' (w (\rneg_{\alpha})^c u \text{ and } u (\rneg_{\beta})^c w' \text{ and } \mM, w' \models \neg\varphi)\\
			& \iff & \exists u( w (\rneg_{\alpha})^c u \text{ and } \exists w'(u (\rneg_{\beta})^c w' \text{ and } \mM, w' \models \neg\varphi))\\
			& \iff & \exists u( w (\rneg_{\alpha})^c u \text{ and } \mM, u\models\neg [\beta]\varphi)\\
			& \iff & \mM, w \models \neg[\alpha][\beta]\varphi
			\end{array}$
		\end{enumerate}
		
			
		We skip 2. as it follows analogous steps.

	\begin{enumerate}
		\item[3.] $\vphantom{\ }$
	
		$\begin{array}{rcl}
			\mM, w \models [\varphi?]\psi & \iff & \forall w'(	w\rpos_{\varphi?}w' \text{ implies } \mM, w'\models \psi)\\
			& \iff & \mM, w\models\varphi \text{ implies } \mM, w\models \psi\\
			& \iff & \mM, w \models \varphi\rightarrow \psi\\
			& & \\
			& & \\
			\mM, w \models \neg[\varphi?]\psi & \iff & \exists w'( w 	(\rneg_{\varphi?})^c w' \text{ and } \mM, w'\models \neg\psi)\\
			& \iff & \mM,w\not\models \neg\varphi \text{ and } \mM, w\models 	\neg\psi\\
			& \iff & \mM, w \models \neg(\varphi\rightarrow \psi)
		\end{array}$
	\end{enumerate}

%
			
	\begin{enumerate}
		\item[4.] $\vphantom{\ }$
		
		$\begin{array}{rcl}
			\mM, w \models [ \alpha^* ] \varphi & \iff & \forall w'( 	w\rpos_{\alpha^*}w' \text{ implies } \mM, w'\models \varphi)\\
			& \iff & (w\rpos_{\alpha^*}w \text{ implies } \mM, w\models \varphi) \text{ and}\\
			& & (\forall w'(\exists u(w\rpos_{\alpha} u \text{ and } u\rpos_{\alpha^*}w') \text{ implies } \mM, w'\models \varphi))\\
			& \iff & \mM, w\models \varphi \text{ and } \forall w'(\exists u(w\rpos_{\alpha} u \text{ and } u\rpos_{\alpha^*}w') \text{ implies } \mM, w'\models \varphi)\\
			& \iff & \mM, w\models \varphi \text{ and } \forall w',u(w(\rpos_{\alpha})^c u \text{ or } u(\rpos_{\alpha^*})^c w' \text{ or } \mM, w'\models \varphi)\\
			& \iff & \mM, w\models \varphi \text{ and } \forall u(w(\rpos_{\alpha})^c u \text{ or } \forall w'(u\rpos_{\alpha^*} w' \text{ implies } \mM, w'\models \varphi))\\
			& \iff & \mM, w\models \varphi \text{ and } \forall u(w\rpos_{\alpha} u \text{ implies } \mM,u\models [\alpha^*]\varphi)\\
			& \iff & \mM, w\models \varphi \text{ and } \mM,w\models [\alpha][\alpha^*]\varphi)\\
			& \iff & \mM, w\models \varphi\wedge [\alpha][\alpha^*]\varphi)\\
			& & \\
			& & \\
			\mM, w \models \neg [ \alpha^* ] \varphi & \iff & \exists w'( w (\rneg_{\alpha^*})^c w' \text{ and } \mM, w'\models \neg\varphi)\\
			& \iff &  (w (\rneg_{\alpha^*})^c w \text{ and } \mM, w\models \neg\varphi) \text{ or}\\
			& & \exists w',w''(w (\rneg_{\alpha})^c w'' \text{ and } w'' (\rneg_{\alpha^*})^c w' \text{ and }\mM, w'\models \neg \varphi)\\
			& \iff &  \mM, w\models \neg\varphi \text{ or } \exists w''(w (\rneg_{\alpha})^c w'' \text{ and } \mM, w''\models \neg[\alpha^*] \varphi)\\
			& \iff &  \mM, w\models \neg\varphi \text{ or } \mM, w\models \neg[\alpha][\alpha^*] \varphi\\
			& \iff & \mM, w \models \neg(\varphi\wedge [\alpha][\alpha^*] \varphi)
		\end{array}$
	\end{enumerate}
		
%
			
%

\end{proof}
	
	Consider the following example:
	
	\begin{ex}\label{ex:dynamic}
		Consider the model $\mM$ represented in Figure~\ref{fig:dynamic}. The following holds:
		
		\begin{figure}[ht]
			\centering
			\includegraphics[width=0.4\textwidth]{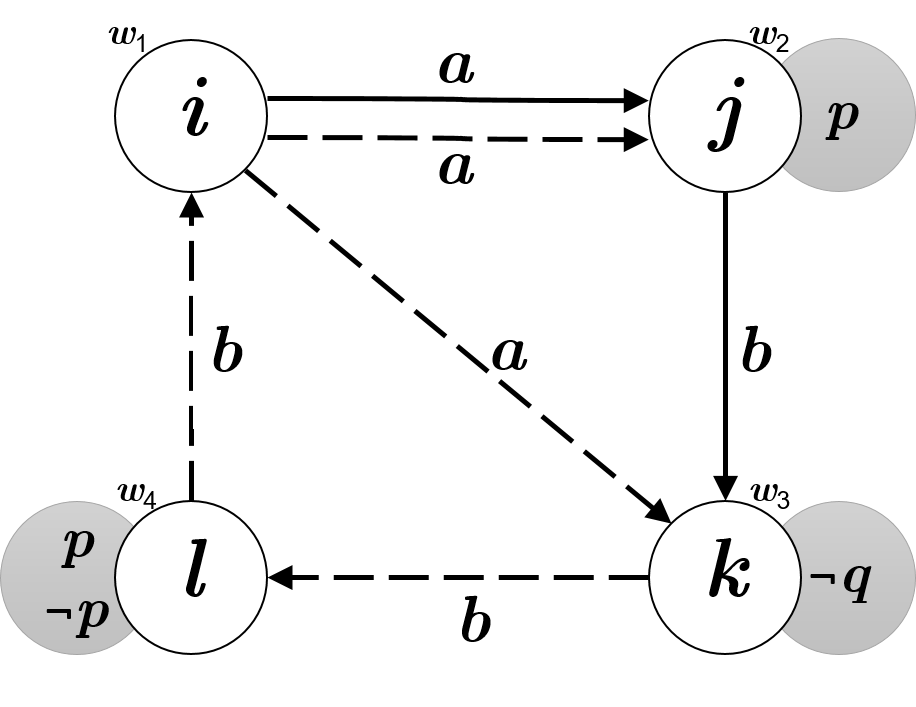}
			\caption{Graphical representation of the model $\mM$.\label{fig:dynamic}}
		\end{figure}
		
		$$\begin{array}{l}
			1.\ \mM \models @_i [a\cup b] p\\
			2.\ \mM \models @_i \neg\langle a;b\rangle (\neg p\wedge q)\\
			3.\ \mM \models @_i \langle a^*\rangle j\\
			4.\ \mM \models @_k \neg[p?] q
		\end{array}$$
		
		The equivalences proved in Lemma~\ref{lem:dynamic_axioms} simplify the previous inferences as follows:
		
		$$\begin{array}{rcl}
			1. & \iff & \mM \models @_i [a] p \wedge @_i [b] p\\
			2. & \iff & \mM \models @_i \neg \langle a\rangle  \langle b\rangle (\neg p \wedge q)\\
			3. & \iff & \mM \models @_i j \vee @_i\langle a\rangle \langle a^* \rangle j\\
			4. & \iff & \mM \models @_k \neg (p\rightarrow q)
		\end{array}$$
	\end{ex}
	
\section{A tableau system for \fdl}
	This section introduces a terminating, sound and complete tableau system for \fdl\ and a decision procedure to verify if a formula is a (global) consequence of a set of formulas. In particular, we can verify if a formula is valid if it is a consequence of the empty set. But first, it is necessary to fix some notation. We start by considering an extra-logical operator $.^-$ applied to formulas and which acts on the satisfaction relation in the following sense: for a model $\mM$, a state $w$ and a formula $\varphi\in\Fm(\lL_{dyn})$,
	
	\[\mM,w \models \varphi^-  \Leftrightarrow \mM, w\not\models \varphi\]
	
	\noindent and, analogously,
	
	\[\mM\models \varphi^- \Leftrightarrow \mM\not\models \varphi.\]
	
	It easy to check that $\mM\models \varphi^-$ if and only if $\exists w\in \Ww: \mM,w\models \varphi^-$. For convenience we will call $\varphi^-$ a \emph{minus-formula}, and the set $\Fm^\bullet(\lL_{dyn})=\Fm(\lL_{dyn})\cup\{\varphi^-\ |\ \varphi\in\Fm(\lL_{dyn})\}$ the set of all signed formulas over $\lL_{dyn}$.
	
	The tableau rules are clustered into two sets. First, we introduce rules for regular and minus-formulas in Figure \ref{fig:rules}.
	Note that the rules directly dealing with modal operators are only applicable if the program is atomic.
	The tableau rules for modal formulas with composite programs can be found in Figure \ref{fig:rules_composite_progs}. Those rules are written in a condensed format which should be interpreted as follows: whenever the premise includes the red and/or blue parts, the conclusion should include the red and/or blue parts accordingly. Take for example the condensed rule $(\textcolor{blue}{\neg}[\alpha;\beta]^{\textcolor{red}{-}})$, from which we can extract the following four rules:
	
	{\small
	\begin{center}
	$\begin{array}{c}
					\displaystyle \frac{@_i [\alpha;\beta]\varphi}{@_i [\alpha][\beta]\varphi}\;([\alpha;\beta])
					\qquad
					\displaystyle \frac{@_i \neg[\alpha;\beta]\varphi}{@_i \neg[\alpha][\beta]\varphi}\;(\neg[\alpha;\beta])
					\\
					\\
					\displaystyle \frac{(@_i [\alpha;\beta]\varphi)^-}{(@_i [\alpha][\beta]\varphi)^-}\;([\alpha;\beta]^-)
					\qquad
					\displaystyle \frac{(@_i \neg[\alpha;\beta]\varphi)^-}{(@_i \neg[\alpha][\beta]\varphi)^-}\;(\neg[\alpha;\beta]^-)
	\end{array}$
	\end{center}
    }

	The rules for the iteration program are an exception; they will play a crucial role in determining the nature of a branch and require a special treatment which shall be discussed in due time. These rules constitute our system, and a tableau here is denoted by $\tau$.

	To put it simply, a tableau consists of a rooted tree, each node of which is essentially a set of signed formulas. The initial (finite) set is constituted of what we call root formulas. A branch is any path in the tree starting at the root. The successors of a node result from the application of a rule to the formulas in the branch to which it belongs. A node without successors is called a leaf. Branching occurs when a rule splits the current branch into two.
	
	The rules $(@_\mathrm{I})$, $(\mathrm{Id})$, $(\mathrm{Nom})$, $([a])$, $(\neg\langle a\rangle)$, $({\langle a\rangle^-})$ and $({\neg[a]^-})$ are called \mbox{\emph{non-destructive}} and the remaining ones are called \emph{destructive}. This distinction is made so that in the systematic tableau construction algorithm a destructive rule is applied at most once to a formula (a destructive rule has exactly one formula in the premise; the converse is not true). As in \cite{hybridproof}, the classification of rules as destructive and non-destructive corresponds to a classification of formulas according to their form.
	
	Rules that introduce new nominals, $(\langle a\rangle), (\neg[a]), (\neg \langle a \rangle ^-)$ and $([a]^-)$, are called existential rules.
	
\begin{figure}[htpb!]
	\begin{minipage}{1.0\textwidth}
	{
	\small
	\noindent$\displaystyle\begin{array}{c}
	\displaystyle\frac{\varphi}{@_i \varphi}\;({@_\mathrm{I}})(\mathrm{i})
	\quad\quad
	\displaystyle\frac{@_i @_j \varphi}{@_j \varphi}\;({@_\mathrm{E}})
	\quad\quad
	\displaystyle \frac{@_i (\varphi\wedge \psi)}{\begin{aligned}
	@_i \varphi\\[-2\jot] @_i \psi\end{aligned}}\;({\wedge})
	\quad\quad
	\displaystyle \frac{@_i (\varphi\vee \psi)}{\begin{array}{c|c}
	@_i \varphi & @_i \varphi \end{array}}\;({\vee})\\\\
	\displaystyle \frac{@_i (\varphi\rightarrow\psi)}{\begin{array}{r|l}(@_i  \varphi)^- & @_i \psi	\end{array}}\;(\rightarrow)
	\quad\quad
	\displaystyle \frac{@_i [a] \varphi, @_i \langle a\rangle  j}{@_j \varphi}\;({[a]})
	\quad\quad
	\displaystyle \frac{@_i \langle a\rangle \varphi}{\begin{array}{c}
	@_i \langle a\rangle  t\\[-0.5\jot] @_t \varphi
	\end{array}}\;({\langle a\rangle})(\mathrm{ii})\\\\
	\displaystyle \frac{@_i\neg @_j\varphi}{@_j\neg\varphi}\;(\neg @)
	\quad\quad
	\displaystyle \frac{@_i\neg (\varphi\wedge\psi)}{\begin{array}{r|l}
	@_i\neg \varphi & @_i \neg\psi \end{array}}\;(\neg\wedge)
	\quad\quad
	\displaystyle \frac{@_i\neg(\varphi\vee\psi)}{\begin{aligned}
	@_i \neg\varphi\\[-2\jot] @_i \neg\psi\end{aligned}}\;(\neg\vee)\\\\
	\displaystyle \frac{@_i \neg(\varphi\rightarrow\psi)}{\begin{array}{c}
	(@_i\neg\varphi)^-\\ @_i\neg\psi\end{array}}\;(\neg{\rightarrow})
	\quad\quad
	\displaystyle\frac{@_i \neg[ a] \varphi}{\begin{array}{c}
	@_i\neg[ a]\neg t\\[-0.25\jot] @_t \neg \varphi\end{array}}\;(\neg[ a])(\mathrm{iii})
	\quad\quad
	\displaystyle \frac{@_i \neg\langle a\rangle \varphi, @_i\neg [ a]\neg j}{@_j\neg\varphi}\;(\neg\langle a\rangle)\\\\
	\displaystyle \frac{@_i \neg\neg\varphi}{@_i \varphi}\;(\neg\neg)
	\quad\quad
	\displaystyle \frac{@_i j,  @_i \varphi}{ @_j\varphi}\;({\mathrm{Nom}})(\mathrm{iv})
	\quad\quad
	\displaystyle \frac{}{@_i i}\;({\mathrm{Id}})(\mathrm{v})\\\\
	\end{array}$
	
	\vspace*{-0.5cm}
	\begin{longtable}[l]{rl}
		(i) & $\varphi$ is not a satisfaction statement, $i$ is in the branch;\\
		(ii) & $\varphi\notin\Nom$, $t$ is a new nominal;\\
		(iii) & $\varphi\neq \neg i$ for all $i\in\Nom$, $t$ is a new nominal;\\
		(iv) & for $@_i \varphi$ a literal;\\
		(v) & for $i$ in the branch.
	\end{longtable}}
\vspace*{-0.5cm}
	\subcaption{Tableau rules for regular formulas.}
	\label{fig:rules_regular}
\end{minipage}

\bigskip

\begin{minipage}{1.0\textwidth}
	{\small
	\centerline{$\begin{array}{c}
	\displaystyle \frac{ \varphi^-}{(@_t\varphi)^- }\;({@^-_\mathrm{I}})(\mathrm{vi}) 
	\quad\quad
	\displaystyle \frac{(@_i @_j \varphi)^-}{(@_j \varphi)^-}\;({@^-_\mathrm{E}})
	\quad\quad
	\displaystyle \frac{(@_i (\varphi\wedge \psi))^-}{\begin{array}{r|l}
	(@_i  \varphi)^- & (@_i \psi)^-	\end{array}}\;({\wedge^-})
	\quad\quad
	\displaystyle \frac{(@_i (\varphi\vee \psi))^-}{\begin{array}{c}
	(@_i  \varphi)^-\\ (@_i \psi)^- \end{array}}\;({\vee^-})\\\\
	\displaystyle \frac{(@_i (\varphi\rightarrow \psi))^-}{
	\begin{array}{c} @_i  \varphi\\ (@_i \psi)^-
	\end{array}}\;({\rightarrow^-})
	\quad\quad
	\displaystyle \frac{(@_i [ a] \varphi)^-}{\begin{array}{c}
	@_i \langle a\rangle t\\[-0.25\jot] (@_t \varphi)^-\end{array}
	}\;({[ a]^-})(\mathrm{vii})
	\quad\quad
	\displaystyle \frac{(@_i \langle a\rangle \varphi)^-, @_i \langle a\rangle  j}{(@_j \varphi)^-}\;({\langle a\rangle^-})\\\\
	\displaystyle \frac{(@_i \neg(@_j\varphi))^-}{(@_j\neg\varphi)^-}\;({\neg @^-})
	\quad\quad
	\displaystyle \frac{(@_i\neg (\varphi\wedge\psi))^-}{\begin{array}{c}
	(@_i\neg \varphi)^-\\ (@_i \neg\psi)^-\end{array}}\;({\neg\wedge^-})
	\quad\quad
	\displaystyle \frac{(@_i \neg(\varphi\vee \psi))^-}{\begin{array}{c|c}
	(@_i \neg\varphi)^- & (@_i \neg\varphi)^- \end{array}}\;({\neg\vee^-})\\\\
	\displaystyle \frac{(@_i\neg(\varphi\rightarrow\psi))^-}{\begin{array}{r|l}
	@_i\neg\varphi & (@_i\neg \psi)^-\end{array}}\;({\neg{\rightarrow}^-})
	\quad\quad
	\displaystyle \frac{(@_i \neg[ a] \varphi)^-, @_i\neg[ a]\neg j}{(@_j\neg  \varphi)^-}\;({\neg[ a]^-})
	\quad\quad
	\displaystyle \frac{(@_i \neg\langle a\rangle \varphi)^-}{
	\begin{array}{c} @_i\neg[ a]\neg t \\ (@_t\neg \varphi)^-
	\end{array}}\;({\neg\langle a\rangle^-})(\mathrm{vii})\\\\
	\displaystyle \frac{(@_i\neg\neg\varphi)^-}{(@_i \varphi)^-}\;({\neg\neg^-})
	\quad\quad
	\displaystyle \frac{(@_i\varphi)^-}{@_i{\neg}\varphi }\;({\mathrm{Id}^-})(\mathrm{viii})\\\\
\end{array}$}
	
	\vspace*{-0.5cm}
	\begin{longtable}[l]{rl}
		(vi) & $\varphi$ is not a satisfaction statement, $t$ is a new nominal;\\
		(vii) & $t$ is a new nominal;\\
		(viii) & $\varphi=j$ or $\varphi=\neg j$, where $j\in\Nom$.
	\end{longtable}}
	\subcaption{Tableau rules for minus-formulas.}
	\label{fig:rules_minus}
\end{minipage}
\caption{Tableau rules. (1/2)}
\label{fig:rules}
\end{figure}

\begin{figure}[ht]
	{\small
		\centerline{$\begin{array}{c}
			\displaystyle \frac{\redminus{@_i \blueneg{[\alpha;\beta]\varphi}}}{\redminus{@_i \blueneg{[\alpha][\beta]\varphi}}}\;(\textcolor{blue}{\neg}[\alpha;\beta]^{\textcolor{red}{-}})
			\quad\quad
			\displaystyle \frac{\redminus{@_i \blueneg{\langle\alpha;\beta\langle\varphi}}}{\redminus{@_i \blueneg{\langle\alpha\rangle\langle\beta\rangle\varphi}}}\;(\textcolor{blue}{\neg}\langle\alpha;\beta\rangle^{\textcolor{red}{-}})\\\\
			\displaystyle \frac{\redminus{@_i \blueneg{[\alpha\cup\beta]\varphi}}}{\redminus{@_i \blueneg{[\alpha]\varphi \wedge [\beta]\varphi}}}\;(\textcolor{blue}{\neg}[\alpha\cup\beta]^{\textcolor{red}{-}})
			\quad\quad
			\displaystyle \frac{\redminus{@_i \blueneg{\langle					\alpha\cup\beta\rangle\varphi}}}{\redminus{@_i \blueneg{\langle\alpha\rangle\varphi \vee \langle\beta\rangle\varphi}}}\;(\textcolor{blue}{\neg}\langle\alpha\cup\beta\rangle^{\textcolor{red}{-}})\\\\
			\displaystyle \frac{\redminus{@_i \blueneg{[\psi?]\varphi}}}{\redminus{@_i \blueneg{\psi\rightarrow \varphi}}}\;(\textcolor{blue}{\neg}[\psi?]^{\textcolor{red}{-}})
			\quad\quad
			\displaystyle\frac{\redminus{@_i \blueneg{\langle\psi?\rangle\varphi}}}{\redminus{@_i \blueneg{\psi\wedge \varphi}}}\;(\textcolor{blue}{\neg}\langle\psi?\rangle^{\textcolor{red}{-}})\\\\
			\multicolumn{1}{l}{\displaystyle \frac{\blueminus{@_i \blueneg{[\alpha^*]\varphi}}}{\begin{array}{c}
			\blueminus{@_i \blueneg{\varphi}}\\
			\blueminus{@_i \blueneg{[\alpha][\alpha^*]\varphi}}\end{array}}\;([\alpha^*]/\textcolor{blue}{\neg} [\alpha^*]^{\textcolor{blue}{-}})(\mathrm{ix})
			\quad
			\displaystyle \frac{@_i \neg[\alpha^*]\varphi}{\begin{array}{c|c}
			@_i \neg\varphi & (@_i \neg \varphi)^-\\
			& @_i \neg[\alpha][\alpha^*]\varphi\end{array}}\;(\neg[\alpha^*])
			\quad
			\displaystyle \frac{(@_i [\alpha^*]\varphi)^-}{\begin{array}{c|c}
			(@_i \varphi)^- & @_i \varphi\\
			& (@_i [\alpha][\alpha^*]\varphi)^-\end{array}}\;([\alpha^*]^-)
			}\\\\
			\multicolumn{1}{c}{
			\displaystyle \frac{@_i \langle\alpha^*\rangle\varphi}{\begin{array}{c|c}
			@_i \varphi & (@_i \varphi)^-\\
			& @_i \langle	\alpha\rangle\langle					\alpha^*\rangle\varphi \end{array}}\;(\langle					\alpha^*\rangle)
			\quad
			\displaystyle \frac{(@_i \neg\langle					\alpha^*\rangle\varphi)^-}{\begin{array}{c|c}
			(@_i \neg\varphi)^- & @_i \neg\varphi\\
			& (@_i \neg\langle	\alpha\rangle\langle					\alpha^*\rangle\varphi)^- \end{array}}\;(\neg\langle					\alpha^*\rangle^-)
			\quad
			\displaystyle \frac{\redminus{@_i \blueneg{\langle					\alpha^*\rangle\varphi}}}{\begin{array}{c}
			\redminus{@_i \blueneg{\varphi}}\\
			\redminus{@_i \blueneg{[\alpha][\alpha^*]\varphi}} \end{array}}\;(\textcolor{blue}{\neg}\langle\alpha^*\rangle/\langle\alpha^*\rangle^{\textcolor{red}{-}})(\mathrm{x})}\\\\
		\end{array}$}
	\vspace*{-0.5cm}
	\begin{longtable}[l]{rl}
	(ix) & either both $\neg$ and $^-$ are present, or neither is present;\\
	(x) & one of $\neg$ or $^-$ is present, but not both.
	\end{longtable}}
	\vspace*{-0.3cm}
	\caption{Tableau rules. (2/2)}
	\label{fig:rules_composite_progs}
\end{figure}

\begin{defn}[Fischer-Ladner closure]\label{def:fladner}
	Let $\varphi\in\Fm^\bullet(\lL_{dyn})$. The \emph{Fischer-Ladner closure} of $\varphi$, denoted $\cl(\varphi)$, is the smallest set of formulas with the following properties:
	\begin{enumerate}
		\item $\varphi\in \cl(\varphi)$ if $\varphi\in\Fm(\lL_{dyn})$;
		\item if $\varphi=\psi^-$ then $\psi\in\cl(\varphi)$;
		\item if $\psi \in \cl(\varphi)$ and $\psi\neq\neg\delta$, then $\neg \psi\in \cl(\varphi)$;
		\item if $\neg\psi, \psi\wedge\delta, \psi\vee\delta, \psi\rightarrow\delta, \langle \alpha\rangle \psi, [\alpha]\psi, @_i\psi \in \cl(\varphi)$ then $\psi,\delta \in \cl(\varphi)$;
		\item if $\langle\alpha;\beta\rangle\psi \in \cl(\varphi)$, respectively
		$[\alpha;\beta]\psi \in \cl(\varphi)$, then $\langle\alpha\rangle\langle\beta\rangle\psi \in \cl(\varphi)$, respectively $[\alpha][\beta]\psi \in \cl(\varphi)$;
		\item if $\langle\alpha\cup\beta\rangle\psi \in \cl(\varphi)$, respectively
		$[\alpha\cup\beta]\psi \in \cl(\varphi)$, then $\langle\alpha\rangle\psi,\langle\beta\rangle\psi \in \cl(\varphi)$, respectively $[\alpha]\psi,[\beta]\psi \in \cl(\varphi)$;
		\item if $\langle\alpha^*\rangle\psi \in \cl(\varphi)$, respectively
		$[\alpha^*]\psi \in \cl(\varphi)$, then $\langle\alpha\rangle\langle\alpha^*\rangle\psi \in \cl(\varphi)$, respectively $[\alpha][\alpha^*]\psi \in \cl(\varphi)$;
		\item if $\langle\delta?\rangle\psi \in \cl(\varphi)$, respectively
		$[\delta?]\psi \in \cl(\varphi)$, then $\varphi\wedge\psi \in \cl(\varphi)$, respectively $\varphi\rightarrow\psi \in \cl(\varphi)$.
		
	\end{enumerate}
\end{defn}

\begin{lem}[Closure property]\label{lem:closure_prop}
	If $@_i\psi\in \tau$, where $\psi\neq j, \langle a\rangle j, \neg[ a]\neg j$ for $ a\in\Act$, $j\in\Nom$, or if $(@_i\psi)^-\in\tau$, then either $\psi\in\cl(\varphi)$ or, in case $\psi=\neg\delta$,  $\delta\in\cl(\varphi)$, where $\varphi$ is a root formula.
\end{lem}

\begin{proof}
	The proof can be obtained by checking each rule.\quad
\end{proof}

The following is a consequence of this result:

\begin{cor}
	For any tableau $\tau$ and nominal $i$, the following sets are finite:
	
	$$\begin{array}{lcl}
		\mathrm{\Gamma}_i & = & \{\varphi\ |\ @_i \varphi \in \tau, \text{ where } \varphi\neq \langle a\rangle j, \neg [a] \neg j, \text{ for }j\in\Nom,a \in \Act \}\\
		\mathrm{\Gamma}_i^- & = & \{\varphi\ |\ (@_i\varphi)^- \in \tau\}
	\end{array}$$
\end{cor}

We define a binary relation between nominals naming the same states and a second one to establish the precedence of nominals as follows:

\begin{defn}\label{def:same_nominal}
	Let $\Theta$ be a branch of a tableau and let $\Nom^\Theta$ be the set of nominals occurring in the formulas of $\Theta$. Define a binary relation $\same$ on $\Nom^\Theta$ by $i \same j$ if and only if the formula $@_i j\in\Theta$.
\end{defn}

\begin{defn}\label{def:precedence}
	Let $i$ and $j$ be nominals occurring on a branch $\Theta$ of a tableau $\tau$. The nominal $i$ is \emph{included} in the nominal $j$ with respect to $\Theta$, denoted $i \precededby j$, if, for any formula $\psi\in\cl(\varphi)$, where $\varphi$ is a root formula, the following holds: 
	\begin{itemize}\setlength\itemsep{0cm}
		\renewcommand\labelitemi{--}
		\item if $\redminus{@_i {\color{blue} \neg}\psi}\in\Theta$, then $\redminus{@_j {\color{blue} \neg}\psi}\in\Theta$, with the red and/or blue parts occurring accordingly in the premise and conclusion; and
		\item the first occurrence of $j$ in $\Theta$ is before the first occurrence of $i$.
	\end{itemize}
\end{defn}

A tableau is built according to the following algorithm:

\begin{defn}[Tableau construction]\label{def:tableau_constr}
	Let $\Delta$ be a finite set of signed formulas in $\Fm^\bullet(\lL_{dyn})$. A \emph{tableau for} $\Delta$ is built inductively according to the following rules:
	
	\begin{itemize}
		\setlength\itemsep{0cm}
		\renewcommand\labelitemi{--}
		\item The one branch tableau $\tau^0$ composed of the formulas in $\Delta$ is a tableau for $\Delta$;
		\item The tableau $\tau^{n+1}$ is obtained from the tableau $\tau^{n}$ if it is possible to apply an arbitrary rule to $\tau^n$ which obeys the following three restrictions:
		\begin{enumerate}
			\setlength\itemsep{0cm}
			\item[$\mathrm{(1)}$] If a formula that result from the application of a rule already occurs in the branch, then its addition is simply omitted;
			\item[$\mathrm{(2)}$] A destructive rule is only applied once to the same formula in each branch;
			\item[$\mathrm{(3)}$] An existential rule is not applied to $@_i \varphi/(@_i\varphi)^-$ (for an appropriate $\varphi$)
			on a branch $\Theta$ if there exists a nominal $j$ such that $i\precededby j$.
		\end{enumerate}
	\end{itemize}
\end{defn}

Therefore a formula cannot occur more than once on a branch, a destructive rule cannot be applied more than once to the same formula in a branch and the introduction of new nominals is controlled.

Termination of the tableau construction algorithm is ensured using the same approach as in \cite{jlamppaper}. A brief explanation follows: start by observing that new nominals are introduced in the tableau by the application of one of the existential rules or the $(@_\mathrm{I}^-)$ rule. The latter, classified as a destructive rule, is applied only once to a minus-formula $\varphi^-$ as long as $\varphi$ is not a satisfaction statement; given that the set $\Delta$ of root formulas is finite, so is the number of times that the rule is applied overall and therefore the number of new nominals generated by this particular rule.

We may classify the elements of $\Nom^\Theta$ according to their origin: a nominal is self-generated if it appears in a root formula or it appears as the product of applying the $(@_\mathrm{I}^-)$ rule; a nominal $j$ is generated by $i$ in $\Theta$ if it appears for the first time after the application of one of the existential rules to a satisfaction formula of the form $@_i \varphi$ or a minus version in the branch (for appropriate $\varphi$). If we denote the top ``source'' of nominals by $\star$ we may consider an analogous ordering $<_\Theta$ on nominals given their origin, as the one from Definition 12 in \cite{jlamppaper}:

\begin{itemize}\setlength\itemsep{0cm}
	\renewcommand\labelitemi{--}
	\item $i <_\Theta j$, with $i,j\in\Nom^\Theta$, if and only if $j$ is generated by $i$ in $\Theta$;
	\item $\star <_\Theta j$, with $j\in\Nom^\Theta$, if and only if $@_j$ appears in a root formula or the nominal $j$ is self-generated;
	\item $x\not<_\Theta \star$, for all $x\in\Nom^\Theta\cup\{\star\}$.
\end{itemize}

Then the graph $(\Nom^\Theta\cup \{\star\},<_\Theta)$ can be showed to be a well-founded (without infinite descending chains), finitely branching tree. From that result, together with the fact that the set $\Delta$ is finite and that for each formula its Fisher-Ladner closure set is finite, termination follows. Details may, once again, be checked on the proof of Theorem 2 in \cite{jlamppaper}, taking into account that the Subformula Property mentioned there must be replaced with the Closure Property (Lemma~\ref{lem:closure_prop} in this paper).

Branches of a tableau $\tau$ can be classified into three categories: open, closed or ignorable; a tableau can be either open or closed. The details follow:

\begin{defn}\label{def:branches}
	A branch $\Theta$ is \emph{closed} if for some nominal $i$, both $@_i \varphi$ and $(@_i\varphi)^-$ appear in $\Theta$, or if either $@_i \neg i$, $@_i \bot$ or $(@_i \neg \bot)^-$ appears in $\Theta$;
	
	A branch $\Theta$ is \emph{ignorable} if it is terminal (\emph{i.e.}, no more rules can be applied to formulas in $\Theta$ according to Definition~\ref{def:tableau_constr}) and it fits one of the four following categories:
	
	\begin{itemize}
		\item ignorable branch of type $\langle\alpha^*\rangle\varphi$: there exists a nominal $i$ such that\linebreak $@_i \langle \alpha^*\rangle \varphi \in \Theta$ and for every nominal $j$, $@_ j \langle \alpha^*\rangle \varphi \in \Theta$ implies $(@_j \varphi)^-\in\Theta$;
		
		\item ignorable branch of type $\neg\langle\alpha^*\rangle\varphi^-$: there exists a nominal $i$ such that\linebreak $(@_i \neg\langle \alpha^*\rangle \varphi)^- \in \Theta$ and for every nominal $j$, $(@_j \neg\langle \alpha^*\rangle \varphi)^- \in \Theta$ implies $@_j \neg \varphi\in\Theta$;
		
		\item ignorable branch of type $[\alpha^*]\varphi^-$: there exists a nominal $i$ such that\linebreak $(@_i [\alpha^*] \varphi)^- \in \Theta$ and for every nominal $j$, $(@_j [\alpha^*] \varphi)^- \in \Theta$ implies $@_j \varphi \in\Theta$; or
		
		\item ignorable branch of type $\neg[\alpha^*]\varphi$: there exists a nominal $i$ such that\linebreak $@_i \neg [\alpha^*] \varphi \in \Theta$ and for every nominal $j$, $@_j \neg [\alpha^*] \varphi \in \Theta$ implies \mbox{$(@_j \neg \varphi)^-\in\Theta$}.
	\end{itemize}

	A branch $\Theta$ is \emph{open} if it is terminal and neither closed nor ignorable.
		
	A tableau is open if it has at least one open branch, otherwise it is closed.
\end{defn}

We characterize a nominal as being terminal if it fits the following definition:

\begin{defn}\label{def:terminal_nominal}
	A nominal $t$ is called \emph{terminal} for $\varphi$ in the branch $\Theta$ if $@_t \varphi\in \Theta$ and there exists a nominal $s$ such that $t\subseteq_\Theta s$.
\end{defn}

The following result is immediate:

\begin{prop}\label{prop:existence_term_nominal}
	Let $\Theta$ be an ignorable branch of type $\mathrm{X}$ (of appropriate form). Then there is a terminal nominal $t$ for $\mathrm{X}$ in $\Theta$.
\end{prop}

A branch $\Theta$ is globally satisfiable in a model $\mM=(\Ww, \famrposact, \famrnegact,\linebreak \nval, \vpos, \vneg)$ if for all $\varphi$ in $\Theta$, $\mM\models \varphi$, and for all $\varphi^-$ in $\Theta$, $\mM\not\models \varphi$. A tableau is globally satisfiable if there is a model where one of its branches is so.

Global satisfiability of a tableau is preserved under the application of tableau rules:

\begin{thm}\label{thm:preserve_sat}
	If $\tau$ is a globally satisfiable tableau, then $\tau'$, obtained by applying one of the tableau rules to $\tau$, is globally satisfiable.
\end{thm}

\begin{proof}
	The proof can be obtained by checking each rule.
	
	The correctness of rules in Figure~\ref{fig:rules} is immediate as they are closely connected with the semantics provided in Definition~\ref{defn:satisfaction}: for any rule $\displaystyle \frac{\Lambda}{\begin{array}{c|c|c}
			\Sigma_1 & \cdots & \Sigma_n
	\end{array}}$, $n\geq1$,
	and any model $\mM$,
	if $\mM\models\Lambda$ then $\mM\models\Sigma_1$ or $\ldots$ or $\mM\models\Sigma_n$,
	where $\Lambda$, $\Sigma_1,\ldots,\Sigma_n \subset \Fm^\bullet(\lL_{dyn})$. To illustrate the proof, take the case of rule $(\rightarrow^-)$:
	
	Assume that $\mM \models (@_i (\varphi\rightarrow\psi))^-$. Then $\mM \not\models @_i (\varphi\rightarrow\psi)$, \emph{i.e.}, for some $w\in \Ww$, $\mM, w \not\models @_i (\varphi\rightarrow\psi)$, and equivalently, $\mM, \nval(i) \not\models \varphi\rightarrow\psi$. Therefore, by Definition~\ref{defn:satisfaction}, $\mM, \nval(i) \models \varphi$ and $\mM, \nval(i) \not\models \psi$. Thus for any $w$, $\mM, w \models @_i\varphi$ and $\mM, w\not \models @_i \psi$. It follows that $\mM \models @_i \varphi$ and $\mM \not \models @_i\psi$. The latter can be rewritten as $\mM \models (@_i \psi)^-$. This concludes the proof that the application of this particular rule does not affect global satisfiability of the tableau. Other rules follow an analogous approach.
	
	As for the rules in Figure~\ref{fig:rules_composite_progs}, they mimic the equivalences from Lemma~\ref{lem:dynamic_axioms}. The only cases worth mentioning are those concerning the iteration program where, when the branch splits, the minus dual of the formula in the left \mbox{sub-branch} is introduced in the right sub-branch. The introduction of this element plays a major role in the definition of ignorable branches, as the reader may have already noticed. This  does not affect in any way the result under appreciation; nonetheless, the proof for the case of $(\langle\alpha^*\rangle)$ is presented next as an example. The other cases are analogous.
	
	Let $\mM$ be a model such that $\mM \models @_i \langle \alpha^* \rangle \varphi$, equivalently $\mM, \nval(i)\models \langle \alpha^* \rangle \varphi$. From Lemma~\ref{lem:dynamic_axioms}, $\mM, \nval(i)\models \varphi\vee \langle \alpha\rangle\langle \alpha^*\rangle \varphi$, which means that $\mM, \nval(i)\models \varphi$ or\linebreak \mbox{$\mM, \nval(i)\models \langle \alpha\rangle\langle \alpha^*\rangle \varphi$}. Put another way, $\mM, \nval(i)\not\models \varphi$ implies \mbox{$\mM, \nval(i)\models \langle \alpha\rangle\langle \alpha^*\rangle \varphi$}. In addition, it is always the case that either $\mM, \nval(i)\models \varphi$ or \mbox{$\mM, \nval(i) \not\models \varphi$}, the latter equivalent to $\mM, \nval(i)\models \varphi^-$. Combining these scenarios, either $\mM\models @_i\varphi$, or	$\mM\models (@_i\varphi)^-$ and \mbox{$\mM \models @_i\langle \alpha\rangle\langle \alpha^*\rangle \varphi$}.
	Thus global satisfiability of the branch $\Theta$ is preserved by the application of the rule $(\langle\alpha^*\rangle)$.
\end{proof}

The next result allows us to safely discard ignorable branches in terminal and globally satisfiable tableaux:

\begin{thm}
	If $\tau$ is a terminal and globally satisfiable tableau, then one of its globally satisfiable branches is not ignorable.
\end{thm}

\begin{proof}
	Let $\tau$ be a terminal and globally satisfiable tableau whose globally satisfiable branches are all ignorable.
	
	Take $\Theta$ to be a branch in $\tau$ which is globally satisfiable and additionally is also ignorable of type $\langle\alpha^*\rangle\varphi$. Let $t$ be the terminal nominal for $\langle\alpha^*\rangle\varphi$ in $\Theta$.
	
	By definition of ignorable branch, $@_t \langle\alpha^*\rangle\varphi, (@_t \varphi)^-\in\Theta$ and there exists a nominal $s$ such that $t\precededby s$.
	Also, since $\tau$ is a terminal tableau, $@_t \langle\alpha\rangle \langle\alpha^*\rangle\varphi \in \Theta$.
	
	Furthermore, and once again because $\tau$ is terminal, there is a branch $\Theta'$ in $\tau$ such that $@_t \langle\alpha^*\rangle\varphi, @_t \varphi \in \Theta'$. The branch $\Theta'$ is clearly not $\langle\alpha^*\rangle\varphi$-ignorable.
	
	If $\Theta'$ were any other type of ignorable branch, we could apply an analogous strategy to infer that there would be yet another branch which was not that type of ignorable either. Since there is a finite number of possible kinds of ignorable branch (due to Lemma~\ref{lem:closure_prop}), there will come to a point where we have in our hands a branch which is not any type of ignorable. Let us assume that to be already the case. Therefore, the branch $\Theta'$ must not be globally satisfiable (1).
	
	Even though the introduction of a new nominal through the application of the $(\langle\alpha\rangle)$-rule on the formula $@_t \langle\alpha\rangle \langle\alpha^*\rangle\varphi \in \Theta$ has been prevented in accordance with the third restriction in the tableau construction algorithm (Definition~\ref{def:tableau_constr}), we can still force it.
	
	Let $\tau'$ be the terminal tableau obtained from $\tau$ by applying the aforementioned existential rule followed by any other allowed rule according to the tableau construction algorithm.
		
	Then $@_t \langle \alpha\rangle n, @_n \langle\alpha^*\rangle\varphi \in \Theta$, where $n$ is a new nominal.
	The branch $\Theta$ is split into two sub-branches: $(\mathrm{L})$ (for \emph{left}) and $(\mathrm{R})$ (for \emph{right}) such that $@_n \varphi \in (\mathrm{L})$ and $(@_n \varphi)^-, @_n \langle\alpha\rangle \langle\alpha^*\rangle\varphi \in (\mathrm{R})$.
	
	If $(\mathrm{L})$ is a globally satisfiable branch in $\mM$, then $\Theta'$ is globally satisfiable in the model $\mM'$ which is such that $\Ww'=\Ww$, $\nval' \overset{t}{\sim} \nval$, where $\nval'(t)=\nval(n)$,
	$(\rpos_\alpha)'= \rpos_\alpha \cup \{(w, \nval'(t))\ |\ (w, \nval(t)) \in \rpos_\alpha  \}$. However, this contradicts (1). Therefore $(\mathrm{L})$ cannot be globally satisfiable. That implies that $(\mathrm{R})$ must be so, given that it derives from $\Theta$ which is globally satisfiable.
	
	Since $\Theta$ is globally satisfiable, there exists a model $\mM$ such that\linebreak \mbox{$\mM, \nval(t)\models \langle\alpha^*\rangle\varphi$}. This means that there exists $n\in\mathbb{N}_0, w\in\Ww$ such that\linebreak $\mM, w\models \varphi$.
	Since $\mM, \nval(t)\not\models \varphi$, $n>0$. Moreover, since $(\mathrm{R})$ is globally satisfiable it follows that $n>1$.
	
	We can again force the application of the $(\langle\alpha\rangle)$-rule in $(\mathrm{R})$. The left sub-branch that is created cannot be globally satisfiable, otherwise $(\mathrm{L})$ would be globally satisfiable too, which is a contradiction. It follows that $n>2$. We can apply this construction until $n>n$, which is again a contradiction.
	
	If $\Theta$ were any other kind of ignorable branch, analogous steps would be used and the same conclusion would arise.
	
	Therefore not all globally satisfiable branches of $\tau$ can be ignorable. So, if $\tau$ is a terminal and satisfiable tableau then there must be a globally satisfiable branch which is not ignorable.	
\end{proof}

Soundness is stated and proved as follows (observe that if $\Delta$ is empty, then the theorem shows that $\varphi$ is valid):

\begin{thm}[Soundness]\label{thm:soundness}
	If the tableau with root $\Delta, \varphi^-$ is closed, then $\varphi$ is a consequence of $\Delta$.
\end{thm}

\begin{proof}
	Assume that the tableau $\tau$ with root $\Delta, \varphi^-$ is closed. Then all of its branches are either closed or ignorable; that means that none of its branches is globally satisfiable.
	By definition, there exists no model $\mM$ such that $\mM\models \Delta$ and $\mM\models \varphi^-$; \emph{i.e.}, for every model $\mM$, either $\mM\not\models \Delta$ or $\mM \models \varphi$. Therefore, for every model $\mM$, $\mM\models \Delta$ implies $\mM \models \varphi$.
\end{proof}

In order to prove completeness, we prove that if a terminal tableau has an open branch $\Theta$, then there exists a model $\mM_\Theta$ where all root formulas are valid.

Suppose that $\Theta$ is an open branch of a terminal tableau. We extract a model from $\Theta$ in the same fashion as in \cite{jlamppaper}, and whose details follow:

Let $\U$ be the subset of $\Nom^\Theta$ that contains every nominal $i$ for which there is no nominal $j$ such that $i \precededby j$. Let $\approx$ be the restriction of $\same$ to $\U$; observe that both $\same$ and $\approx$ are equivalence relations. Note also that $\U$ contains all nominals $i$ such that $\star <_\Theta i$.

Given a nominal $i$ in $\U$, we let $[i]_{\approx}$ denote the equivalence class of $i$ with respect to $\approx$ and we let $\U_{\hspace*{-0.5mm}/\hspace*{-0.5mm}\approx}$ denote the set of equivalence classes.

We let $\rpos_a$ be the binary relation on $\U$ defined by $i\rpos_a j$ if and only if there exists a nominal $j'\approx j$ such that one of the following conditions is satisfied:
\begin{enumerate}\setlength\itemsep{0cm}
	\item $@_i\langle a \rangle j'\in \Theta$; or
	\item there exists a nominal $k\in\Nom^\Theta$ such that $@_i\langle a\rangle k\in\Theta$ and $k \precededby j'$.
\end{enumerate}

On the other hand, we define the complement of the binary relation $\rneg_a$ on $\U$ as $i(\rneg_a)^c j$ if and only if there exists a nominal $j'\approx j$ such that one of the following conditions is satisfied:
\begin{enumerate}\setlength\itemsep{0cm}
	\item $@_i\neg[a]\neg j'\in \Theta$; or
	\item there exists a nominal $k\in\Nom^\Theta$ such that $@_i\neg[a]\neg k\in\Theta$ and $k \precededby j'$.
\end{enumerate}

Note that the nominal $k$ referred to in the second items is not an element of $\U$. It follows from $\Theta$ being closed under the rule $(\mathrm{Nom})$ that $\rpos_a$ and $(\rneg_a)^c$ are compatible with $\approx$ in the first argument and it is trivial that they are compatible with $\approx$ in the second argument. We let $\overline{\rpos_a}$, respectively $\left(\overline{\rneg_a}\right)^c$ (observe that once again we define the complement of the relation), be the binary relation on $\U_{\hspace*{-0.5mm}/\hspace*{-0.5mm}\approx}$ defined by $[i]_{\approx}\,\overline{\rpos_a}\,[j]_{\approx}$, respectively $[i]_{\approx}\,\left(\overline{\rneg_a}\right)^c\,[j]_{\approx}$, if and only if $i\rpos_a j$, respectively $i(\rneg_a)^c j$.

Let $\overline{\nval}:\U\rightarrow \U_{\hspace*{-0.5mm}/\hspace*{-0.5mm}\approx}$ be defined such that $\overline{\nval}(i)=[i]_{\approx}$.

Let $\vpos$ be the function that to each ordinary propositional variable assigns the set of elements of $\U$ where that propositional variable occurs, \emph{i.e.}, \mbox{$i\in \vpos(p)$} if and only if $@_i p\in\Theta$.
On the other hand, let $\vneg$ be the function that to each ordinary propositional variable assigns the set of elements of $\U$ where the negation of that propositional variable occurs, \emph{i.e.}, $i\in \vneg(p)$ if and only if $@_i \neg p\in\Theta$. We let $\overline{\vpos}$ be defined by $\overline{\vpos}(p)=\{[i]_{\approx}\ |\ i\in \vpos(p)\}$. We define $\overline{\vneg}$ analogously.

Given a branch $\Theta$, let $\mM_\Theta = \left(U_{\hspace*{-0.5mm}/\hspace*{-0.5mm}\approx},\left(\overline{\rpos_a}\right)_{a\in\Act},\left(\overline{\rneg_a}\right)_{a\in\Act},\overline{\nval}, \overline{\vpos}, \overline{\vneg}\right)$.

We will omit the reference to the branch in $\mM_\Theta$ if it is clear from the context.

\begin{thm}[Model Existence]\label{thm:model_existence}
	Let $\Theta$ be an open branch of a terminal tableau $\tau$. The model extracted from the branch, $\mM$, is such that the following conditions hold:
	
	\begin{enumerate}\setlength\itemsep{0cm}
		\item[$\mathrm{(i)}$] if $@_i\varphi\in\Theta$, then $\mM, [i]_{\approx}\models \varphi$;
		
		\item[$\mathrm{(ii)}$] if $(@_i\varphi)^-\in\Theta$, then $\mM, [i]_{\approx}\not\models \varphi$.
	\end{enumerate}
	
	\noindent whenever $@_i\varphi$ contains only nominals from $\U$.
\end{thm}

\begin{proof}
	First, consider the cases where $@_i \varphi$ is a hybrid formula.
	The proof is done by induction on the complexity of $\varphi$. 
	
	The base step concerns the cases where $\varphi$ is either a nominal, a propositional variable, $\bot$, or the negation of one of those. Satisfaction of $\varphi$ in $\mM, [i]_{\approx}$ is trivially obtained by the construction of the model $\mM$.
	
	We assume that the result holds for subformulas of $\varphi$ and their negation (I.H.). We illustrate the cases $\varphi = \psi\wedge\delta$, $\varphi= \neg (\psi\wedge\delta)$ and $\varphi = \langle a\rangle \psi$, $\varphi = \neg\langle a\rangle \psi$, for $a\in\Act$:
	
	\begin{itemize}
		\item $\varphi=\psi\wedge\delta$
		\begin{itemize}\setlength\itemsep{0cm}
			\item[$\mathrm{(i)}$] $@_i (\psi\wedge\delta)\in\Theta$, then, by applying the rule ($\wedge$), $@_i\psi, @_i\delta\in\Theta$.
			By I.H., $\mM, [i]_{\approx}\models \psi$ and $\mM, [i]_{\approx}\models \delta$. Therefore \mbox{$\mM, [i]_{\approx}\models \psi\wedge\delta$}.
			
			\medskip
			
			\item[$\mathrm{(ii)}$] $(@_i (\psi\wedge\delta))^-\in\Theta$, then, by applying the rule (${\wedge^-}$), $(@_i\psi)^-\in \Theta$ or $(@_i\delta)^-\in\Theta$.
			Hence, by I.H., $\mM, [i]_{\approx}\not\models \psi$ or \mbox{$\mM, [i]_{\approx}\not\models \delta$}. Therefore, $\mM, [i]_{\approx}\not\models \psi\wedge \delta$.
		\end{itemize}
		
		\medskip
		
		\item $\varphi=\neg(\psi\wedge\delta)$
		\begin{itemize}\setlength\itemsep{0cm}
			\item[$\mathrm{(i)}$] $@_i \neg(\psi\wedge\delta)\in\Theta$, then, by applying the rule ($\neg\wedge$), $@_i\neg\psi\in\Theta$ or \mbox{$@_i\neg\delta\in\Theta$}.
			By I.H., either $\mM, [i]_{\approx}\models \neg\psi$ or $\mM, [i]_{\approx}\models \neg\delta$. Therefore, \mbox{$\mM, [i]_{\approx}\models \neg(\psi\wedge\delta)$}.
			
			\medskip
			
			\item[$\mathrm{(ii)}$] $(@_i \neg(\psi\wedge\delta))^-\in\Theta$, then, by applying the rule (${\neg\wedge^-}$), $(@_i\neg\psi)^-,\linebreak (@_i\neg\delta)^-\in\Theta$.
			Hence, by I.H., $\mM, [i]_{\approx}\not\models \neg\psi$ and $\mM, [i]_{\approx}\not\models \neg\delta$. Therefore, $\mM, [i]_{\approx}\not\models \neg(\psi\wedge\delta)$.
		\end{itemize}
	
		\medskip
	
		\item $\varphi=\langle a\rangle \psi$
		\begin{itemize}\setlength\itemsep{0cm}
			\item[$\mathrm{(i)}$]
			
			\begin{itemize}\setlength\itemsep{0cm}
				\item if $\psi=j$, $j\in\Nom$:
				
				$@_i\langle a\rangle j\in\Theta$, then $[i]_{\approx}\overline{\rel_a^+}[j]_{\approx}$ and by definition \mbox{$\mM,[i]_{\approx}\models \langle a \rangle j$}.
				
				\item if $\psi$ is not a nominal:
				
				$@_i\langle a\rangle \psi\in\Theta$, then by the application of the rule $(\langle a \rangle)$, $@_i\langle a\rangle t$ and $@_t\psi\in\Theta$, for a new nominal $t$. Then:
				
				\begin{itemize}\setlength\itemsep{0cm}
					\item if $t\in \U$, $[i]_{\approx}\overline{\rel_a^+}[t]_{\approx}$. By I.H., $\mM,[t]_{\approx}\models\psi$, so $\mM,[i]_{\approx}\models\langle a\rangle \psi$.
					
					\item if $t\notin \U$, $\exists s$ such that $t\precededby s$. Assume that there is no $r$ such that $s\precededby r$, \emph{i.e.}, $s\in \U$.
					Then by the definition of $\precededby$, $@_s\psi\in\Theta$. Observe also that $\psi$ is a subformula of $\varphi$, thus by I.H., \mbox{$\mM,[s]_{\approx}\models \psi$}.
					By the construction of $\mM$, $[i]_{\approx}\overline{\rel_a^+}[s]_{\approx}$.
					Therefore, $\mM,[i]_{\approx}\models\langle a\rangle \psi$.
				\end{itemize}
			\end{itemize}
			
			\medskip
			
			\item[$\mathrm{(ii)}$]
			$(@_i\langle a\rangle \psi)^-\in\Theta$. We want to prove that $\mM, [i]_{\approx}\not\models\langle a\rangle \psi$, \emph{i.e.}, that for all $[k]_{\approx}$ such that $[i]_{\approx}\overline{\rel_a^+}[k]_{\approx}$, $\mM, [k]_{\approx}\not\models \psi$.
			
			Let $k$ be a nominal such that $[i]_{\approx}\overline{\rel_a^+}[k]_{\approx}$; by definition there exists $k'$ with $k'{\approx} k$ that satisfies one of the following conditions:
			
			\begin{itemize}\setlength\itemsep{0cm}
				\item $@_i\langle a\rangle k'\in\Theta$, which then by applying the rule (${\langle a\rangle^-}$), implies that $(@_{k'} \psi)^-\in\Theta$. By I.H., $\mM, [k']_{\approx}\not \models \psi$. Since $[k']_{\approx}=[k]_{\approx}$, follows that $\mM, [k]_{\approx}\not\models\psi$. Or:
				
				\item $\exists s\in \Nom^\Theta$ such that $@_i\langle a\rangle s\in\Theta$ and $s\precededby k'$, then by applying ($\langle a\rangle^-$) it follows that $(@_s \psi)^-\in\Theta$. Since $s\precededby k'$, $(@_{k'}\psi)^-\in\Theta$. By I.H., $\mM, [k']_{\approx}\not\models \psi$. Since $[k']_{\approx}=[k]_{\approx}$, then $\mM, [k]_{\approx}\not\models\psi$.
			\end{itemize}
			
			It follows that $\mM, [i]_{\approx}\not\models \langle a\rangle \psi$.
			
		\end{itemize}
		
		\medskip
		
		\item $\varphi=\neg\langle a\rangle\psi$
		
		\begin{itemize}\setlength\itemsep{0cm}
			\item[$\mathrm{(i)}$] $@_i\neg\langle a\rangle\psi\in\Theta$. We want to prove that $\mM, [i]_{\approx}\models\neg\langle a\rangle\psi$, \emph{i.e.}, that for all $[k]_{\approx}$ such that $[i]_{\approx}\left(\overline{\rel_a^-}\right)^c[k]_{\approx}$, $\mM, [k]_{\approx}\models\neg \psi$.
			
			Let $k$ be a nominal such that $[i]_{\approx}\left(\overline{\rel_a^-}\right)^c[k]_{\approx}$;
			by definition, there exists $k'$ with $k'{\approx} k$ that satisfies one of the following two conditions:
			
			\begin{itemize}\setlength\itemsep{0cm}
				\item $@_i\neg[a]\neg k'\in\Theta$ which then by applying the rule ($\neg\langle a\rangle$) implies that $@_{k'} \neg \psi\in\Theta$ and by I.H., $\mM, [k']_{\approx}\models\neg \psi$. Since $[k']_{\approx}=[k]_{\approx}$, follows that $\mM, [k]_{\approx}\models \neg\psi$. Or:
				
				\item $\exists s\in \Nom^\Theta$ such that $@_i\neg[a]\neg s\in\Theta$ and $s\precededby k'$, then by applying the rule ($\neg \langle a\rangle$), $@_s \neg \psi\in\Theta$. Since $s\precededby k'$, $@_{k'}\neg\psi\in\Theta$. By I.H., $\mM, [k']_{\approx}\models\neg \psi$. Since $[k]_{\approx}=[k']_{\approx}$, $\mM, [k]_{\approx}\models\neg \psi$.
			\end{itemize}
			Therefore \mbox{$\mM, [i]_{\approx}\models \neg\langle a\rangle\psi$.}
			
			\medskip
			
			\item[$\mathrm{(ii)}$] $(@_i\neg\langle a\rangle \psi)^-\in\Theta$, thus by applying rule (${\neg\langle a\rangle^-}$), \mbox{$@_i\neg[a]\neg t, (@_t \neg \psi)^-\in\Theta$}, for a new nominal $t$. Then:
			
			\begin{itemize}\setlength\itemsep{0cm}
				\item if $t\in \U$, $[i]_{\approx}\left(\overline{\rel_a^-}\right)^c[t]_{\approx}$. By I.H., $\mM, [t]_{\approx}\not\models \neg\psi$. Thus $\mM,[i]_{\approx}\not\models \neg\langle a\rangle\psi$.
				
				\item if $t\notin \U$, $\exists s$ such that $t\precededby s$. Assume that there is no $r$ such that $s\precededby r$, \emph{i.e.}, $s\in \U$. Since $t\precededby s$, follows that $(@_s\neg\psi)^-\in\Theta$.	By I.H., $\mM,[s]_{\approx}\not\models \neg\psi$ and by definition $[i]_{\approx}\left(\overline{\rel_a^-}\right)^c[s]_{\approx}$. It follows that \mbox{$\mM,[i]_{\approx}\not\models\neg\langle a\rangle \psi$}.
			\end{itemize}
	\end{itemize}
\end{itemize}
	
	Other cases follow an analogous approach.\footnote{For the interested reader, the proof of Theorem 4 in \cite{jlamppaper} may serve as a guide for the remaining cases. Observe that the cases for modal formulas (in this paper we proved the $\langle a\rangle$ case) and implication deviate slightly given the different semantics and, of course, rules provided in that paper for these particular cases.}
	
	\medskip
	
	We present a detailed proof for the cases where $\varphi$ is a modal formula of the form $\langle\alpha\rangle\psi$ or $[\alpha]\psi$. The proof is by induction on the complexity of $\alpha$.
	
	The base cases, where $\alpha$ is an atomic program, reduces to a hybrid formula and thus to one of the cases aforementioned.	
	
	\textbf{Induction Hypothesis} (I.H.): the result holds for formulas $\langle\alpha\rangle\delta, \langle\beta\rangle\delta, [\alpha]\delta$, $[\beta]\delta$ and $\delta$ in the Fischer-Ladner closure of $\varphi$ (Definition \ref{def:fladner}).
	
	We prove only the case $\mathrm{(i)}$ as the other is analogous.
	
	\begin{itemize}
		\item $\varphi=\langle \alpha;\beta\rangle\psi$
		\begin{itemize}\setlength\itemsep{0cm}
			\item[$\mathrm{(i)}$] $@_i \langle \alpha;\beta\rangle\psi\in\Theta$, then, by applying the rule ($\langle \alpha;\beta\rangle$), $@_i \langle \alpha\rangle\langle \beta\rangle \psi\in\Theta$. By I.H. on $@_i \langle\alpha\rangle \delta$, where $\delta=\langle\beta\rangle\psi$, follows that $\mM, [i]_{\approx}\models \langle \alpha\rangle\langle\beta\rangle\psi$; equivalently, $\mM, [i]_{\approx}\models \langle \alpha;\beta\rangle\psi$.
		\end{itemize}
	
	\medskip
	
	\item The cases $\varphi= [\alpha;\beta]\psi, \langle\alpha\cup\beta\rangle\psi, [\alpha\cup\beta]\psi$, the negations of these and of the previous bullet are similar.
	
	\medskip
	
	\item $\varphi=\langle\chi?\rangle \psi$
	\begin{itemize}
		\item[$\mathrm{(i)}$]  $@_i \langle \chi?\rangle \psi\in\Theta$, then by applying the rule $\langle \chi?\rangle$, $@_i (\chi\wedge \psi)\in\Theta$. By I.H.,  $\mM, [i]_{\approx}\models \chi\wedge\psi$. Therefore $\mM, [i]_{\approx}\models \langle \chi?\rangle\psi$.
	\end{itemize}

	\medskip

	\item $\varphi=[\chi?]\psi$ is analogous, as are the cases with negation occurring before the modal operator.
	
	\medskip
	
	\item $\varphi=\langle\alpha^*\rangle\psi$
	\begin{itemize}
		\item[$\mathrm{(i)}$]  $@_i \langle\alpha^*\rangle\psi \in \Theta$, then by applying the rule $(\langle\alpha^*\rangle)$, either $@_i\psi\in\Theta$ or $(@_i \psi)^-$, $@_i\langle\alpha\rangle\langle\alpha^*\rangle\psi\in\Theta$. By I.H., either $\mM, [i]_{\approx}\models \psi$ or\linebreak \mbox{$\mM, [i]_{\approx}\models \langle\alpha\rangle\langle\alpha^*\rangle\psi$}. Anyway, $\mM, [i]_{\approx}\models \langle\alpha^*\rangle\psi$.
	\end{itemize}
	
	\medskip
	
	\item $\varphi=[\alpha^*]\psi$ and versions with negation occurring before the modal operator are once again analogous.
	
	\end{itemize}
\end{proof}

Completeness can be easily proved now:

\begin{thm}[Completeness]\label{thm:completeness}
	If $\varphi$ is a consequence of $\Delta$, then the tableau with root $\Delta,\varphi^-$ is closed.
\end{thm}

\begin{proof}
	Assume that the tableau $\tau$ with root $\Delta,\varphi^-$ is not closed.
	Then $\tau$ is such that it has an open branch, which by Theorem~\ref{thm:model_existence} implies that there is a model $\mM$ where all formulas in $\Delta$ and $\varphi^-$ are globally satisfiable.
	That means that $\varphi$ is not globally satisfiable in $\mM$. Therefore $\varphi$ is not a consequence of $\Delta$.
\end{proof}

\section{Conclusion}
The purpose of this paper is to, first and foremost, start a discussion about the roles of negation and modal operators when combined in a paraconsistent environment. Namely, it begins by presenting a four-valued Hybrid Logic, posteriorly used as the basis of a dynamic version, where the formal duality between modal operators is not preserved. One of the reasons behind this choice stems from the discussions of Fitting in \cite{melv}, for the predicate abstraction of modal formulas in models with varying domains. His ideas are, in some ways, related to the discussion about the fact that negation is always carried inside modal formulas by making use of the formal duality between $\Box$ and $\Diamond$. However, as debated in the present paper, by doing so the negation concerns only propositional variables and/or nominals and completely ignores the transition involved in the semantics of the \emph{whole} modal formula. In cases where accessibility relations are classical this is a natural construction, but it fails to grasp the true meaning of the assertion when the accessibility relations become four-valued functions.

The paraconsistent versions of Modal Logic available in literature and where transitions are many-valued have been thought in ways that they keep formal duality between modal operators. In some cases, such as Modal Bilattice Logic \cite{Riv1}, an extension with nominals and the satisfaction operator is such that $\rel(w,w')\neq \val(@_i\Diamond j, -)$ ($i$ names the state $w$ and $j$ the state $w'$), which is highly undesirable.

This paper aimed to provide a logic where inconsistent models can be \linebreak(uniquely) described resorting to diagrams, \emph{i.e.}, the set of all atomic formulas that hold in the model. This can be achieved when all states in the model are named -- which justifies the use of nominals -- and there is special machinery to be able to refer to each state in the model -- the satisfaction operator $@$. Inconsistent models consist of sets of states together with information about transitions and the absence of transitions between states and information about the verification and refutation of local properties, hence our models incorporate both four-valued relations and four-valued (local) valuations. A four-valued Hybrid Logic, \fhl, is the central point of the first section of the paper. 

The second section introduces a dynamic version of \fhl. The concept of modality is replaced by that of program, which can be composed. One can view each relation $\rel_\alpha$ as the set of pairs of states $(w,w')$ such that $w$ is an input state and $w'$ an output state of the program $\alpha$. From a four-valued perspective, $w\rpos_\alpha w'$ if we have information that the program $\alpha$ runs in $w$ and terminates at $w'$ and $w\rneg_\alpha w'$ if we have information that the program $\alpha$ does not run between $w$ and $w'$. The four-valued approach is a means to deal with information obtained by collecting different testimonies. The interpretations of $\rpos_{\alpha;\beta}, \rpos_{\alpha\cup\beta}, \rpos_{\alpha^*}, \rpos_{\varphi?}$ are analogous to those in Propositional Dynamic Logic; for the negative cases, the interpretations resort to the complements of the underlying relations. The axioms of \pdl\ hold in \fdl, making the latter a truly Dynamic Logic, albeit four-valued and with non-dual modal operators.

The third and final section of the paper introduces a terminating, sound and complete tableau system for \fdl. Some of the techniques in \cite{jlamppaper} were repurposed. But there were also some necessary additions, namely the introduction of the notions of ignorable branch in order to deal with the iteration program and a result to ensure that those branches make their own name justice and can be indeed safely ignored.

Even though \fdl\ was created independently, it would be interesting to check, in the future, its connection to a logic obtained by applying the technique in \cite{madeiradyn}. Another topic for further research is the study of bisimulation.

\bibliographystyle{abbrv}
\bibliography{biblio}

\end{document}